\documentclass[journal]{IEEEtran}
\usepackage{amsfonts,amssymb}
\usepackage{amsthm}
\usepackage{graphicx,graphics}
\usepackage{cite,url}
\usepackage{smartref}
\usepackage{amsmath}
\usepackage{xcolor}
\usepackage{array}
\usepackage{tikz}
\usepackage{arydshln}

\newtheorem{definition}{Definition}
\newtheorem{theorem}{Theorem}
\newtheorem{proposition}{Proposition}
\newtheorem{rem}{Remark}
\allowdisplaybreaks 

\begin{document}
	\title{Single-Anchor Two-Way Localization Bounds for 5G mmWave Systems}
	\author{Zohair Abu-Shaban~\IEEEmembership{Senior Member, IEEE}, Henk Wymeersch~\IEEEmembership{Senior Member, IEEE},\\ 
Thushara Abhayapala~\IEEEmembership{Senior Member, IEEE}, and Gonzalo Seco-Granados~\IEEEmembership{Senior Member, IEEE}
	}

\maketitle
\thispagestyle{empty}
{{\let\thefootnote\relax\footnotetext{
Copyright (c) 2015 IEEE. Personal use of this material is permitted. However, permission to use this material for any other purposes must be obtained from the IEEE by sending a request to pubs-permissions@ieee.org.

Zohair Abu-Shaban is with the School of Engineering and Information Technology, University of New South Wales (UNSW), Canberra, Australia. Email: zohair.abushaban@ieee.org. Henk Wymeersch is with the Department of Signals and Systems, Chalmers University of Technology, Sweden. Email: henkw@chalmers.se. Thushara Abhayapala is with the Research School of Engineering (RSEng), the Australian National University (ANU), Canberra, Australia. Email: thushara.abhayapala@anu.edu.au.  Gonzalo Seco-Granados is with the Department of Telecommunications and Systems Engineering, Universitat Aut\`onoma de Barcelona, Spain (UAB). Email: gonzalo.seco@uab.cat.

Part of this paper was presented in IEEE GLOBECOM Workshops, Abu Dhabi, UAE, Dec 2018.  

This work is supported in part by the Australian Government's Research Training Program (RTP), the Horizon2020 projects HIGHTS  (High precision positioning for cooperative ITS applications) MG-3.5a-2014-636537 and 5GCAR (Fifth Generation Communication Automotive Research and innovation), the VINNOVA COPPLAR project, funded under Strategic Vehicle Research and Innovation Grant No. 2015-04849, and the Spanish Ministry of Economy, Industry and Competitiveness under Grants TEC2017-89925-R and TEC2017-90808-REDT.
}}
\begin{abstract}
	Recently, millimeter-wave (mmWave) 5G localization has been shown to be to provide centimeter-level accuracy, lending itself to many location-aware applications, e.g., connected autonomous vehicles (CAVs).  One assumption usually made in the investigation of localization methods is that the user equipment (UE), {i.e., a CAV,} and the base station (BS) are {time} synchronized. In this paper, we remove this assumption and investigate two two-way localization protocols: (i) a round-trip localization protocol (RLP), whereby the BS and UE exchange signals in two rounds of transmission and then localization is achieved using the signal received in the second round; (ii) a collaborative localization protocol (CLP), whereby localization is achieved using the signals received in the two rounds. We derive the position and orientation error bounds applying beamforming at both ends and compare them to the traditional one-way localization. Our results show that mmWave localization is mainly limited by the angular rather than the temporal estimation and that CLP significantly outperforms RLP. Our simulations also show that it is more beneficial to have more antennas at the BS than at the UE.
\end{abstract}
\section{Introduction}
The fifth generation of mobile communication (5G) using millimeter-wave technology (mmWave) will be the first generation to integrate the location information in the network design and optimization \cite{Taranto2014, Akyildiz2016}, for example through, beamforming \cite{Aviles2016}, pilot assignment \cite{Akbar2016}, and resource allocation \cite{Muppirisetty2016}.

Localization error in mmWave 5G has been shown to be in the order of centimeters, making location-aware applications in 5G much more attractive than ever before. Such applications including targeted content delivery \cite{Ma2014}, vehicular communication \cite{Garcia2016}, and assisted living systems \cite{Witrisal2016}. Of particular interest are systems of connected autonomous vehicles (CAVs) \cite{WhitePaper}, which are a typical use case of 5G communication \cite{Ericsson}, and air-ground communication with unmanned aerial vehicles (UAVs) \cite{Qiu2019}.

Due to the deployment of arrays with a high number of antennas at the transmitter and the receiver, and the utilization of {large} bandwidth \cite{Andrews2014, Pi2011, Rappaport2013, Heath2016, Orhan2015}, localization with a single base station (BS) can be seen as the ultimate localization strategy for 5G. With the high number of antennas, the directions of arrival (DOA) and departure (DOD) can be estimated with {a} very low error \cite{Larsen2009}, while the large bandwidth enables a highly accurate estimation of the time of arrival (TOA) \cite{Shen2007,Shen2010,Shen2010_2,Shen2010_3}, i.e., {a} low-error range estimate. Subsequently, combining the spatial and temporal estimates, the user equipment (UE) location\footnote{In this paper, we use the terms \textit{location/localization} and \textit{position/positioning} interchangeably.} can be estimated. {On the other hand, some papers consider mmWave channels estimation in the beamspace \cite{Huang2019,Gao2019, Fan2017}, so in principle, the AOA and AOD can be deduced directly from the channel estimate. However, the estimation in the beamspace does not show how to estimate the TOA.}

Recently, the accuracy of single-anchor\footnote{In mobile networks, anchor refers to the BS, whose position and orientation are known.} localization for 5G mmWave systems has been studied in several papers in terms of position (PEB) and orientation error bounds (OEB). {PEB and OEB are theoretical bounds that are used to benchmark location estimation techniques, and hence they are measures of the optimality of such techniques.} In \cite{Shahmansoori2017}, the UE PEB and OEB of 2D localization  were investigated using {uniform linear arrays} in 5G mmWave systems. Moreover, \cite{Guerra2017} and \cite{Zohair2017}  derived, with different approaches, the PEB and OEB for mmWave 3D localization using arrays with arbitrary geometry. The results in \cite{Shahmansoori2017, Guerra2017, Zohair2017} showed a 5G mmWave localization performance with an error in the order of centimeters. However, one important, yet usually overlooked, requirement for localization is the synchronization of BS and UE. For example, \cite{Shahmansoori2017} and \cite{Zohair2017} assume that the BS and UE are perfectly synchronized, while \cite{Guerra2017} assumes coarse synchronization, and includes a residual synchronization error in their localization model. Synchronization can be avoided by the use of two-way ranging methods \cite{Sahinoglu2008, Pelka2017, Joon2002}, where the time-of-flight is utilized to estimate the range and clock bias, or three-way ranging \cite{Sahinoglu2008} and multi-way ranging \cite{Duisterwinkel2017, Sark2015} to additionally estimate higher-order artifacts such as clock drift and skew. However, such methods have not been evaluated for mmWave systems. {Such systems possess different features, including highly sparse channels and directional transmission, making the estimation of the angles of arrival and departure as relevant to localization as the time of arrival. Our work is the first to consider such a scenario and investigate the associated two-way positioning performance that is a function of the spatio-temporal properties of the channel.}

{In this paper, we propose two-way localization (TWL), {whereby a known signal is transmitted from the first device, the BS or the UE, to the second device that, in turn, responds by sending another known signal, after which the relative location and orientation of the devices can be estimated.} We study the PEB and OEB under line-of-sight (LOS) communication for two protocols: (i) \textit{Round-trip Localization Protocol (RLP),} where the second device waits for a pre-agreed interval, from the time the first signal is received, before sending another signal to the first device, upon which localization is based; and (ii) \textit{Collaborative Localization Protocol (CLP),} where the second device sends back the received signal to the first device, and localization is based on both signals. {By their nature, these bounds are theoretical and serve as a means to determine performance benchmarks to assess location estimation techniques, to design localization systems, and to determine when the location and orientation can be potentially estimated.}  Our main contributions are:} 
{\begin{itemize}
\item Introducing RLP and CLP for LOS 5G mmWave signals and their analysis in terms of the localization bounds. 
\item For the two protocols, we derive the Fisher information matrices (FIMs) of the position and orientation, and consequently the PEB and OEB, with the timing bias between the BS and UE as a nuisance parameter. 
\item We investigate the impact of the number of antennas at BS and UE, as well as the bandwidth, and show that, in contrast to the standard two-way ranging methods \cite{Sahinoglu2008, Pelka2017, Joon2002}, the TWL performance in mmWave multiple-input multiple-output (MIMO) systems depends on the device that initiates the protocol. 
	\end{itemize}
}
\begin{figure}[!t]
	\centering
	\includegraphics[scale=0.8]{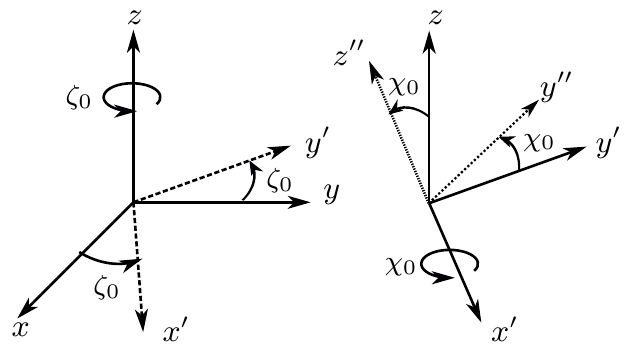}
	\caption{Two-step rotation: {The first rotation is around the $z$-axis, creating $x'$ and $y'$ axes. The second rotation is around $x'$, creating  $y''$ and $z''$ axes.}}
	\label{fig:rotation}
\end{figure}

The initial results of the RLP were presented briefly in \cite{Zohair_DLP}, while in this paper, while in this paper, we i) discuss RLP in more detail, ii) provide CLP as additional protocol, and iii) present more in-depth performance analysis and insightful results on both protocols.

The rest of the paper is organized as follows. The system model, including the considered geometry, channel model and beamforming, is described in Section II, while the proposed protocols are outlined in Section III. Subsequently, FIM basics are introduced at the outset of Section IV, before proceeding to derive the PEB and OEB for both protocols. The numerical simulation results are given in Section V, while the conclusions are highlighted in Section VI.
\section{System Model}\label{sec:sys_model}
\subsection{System Geometry}
Consider a BS located at the origin of the 3D space with zero-orientation angles, and a UE located at a fixed unknown position $\mathbf{p}\triangleq{[p_x,p_y,p_z]^\mathrm{T}}$ with unknown orientation angles $\mathbf{o}\triangleq{[\zeta_0,\chi_0]^\mathrm{T}}$. As illustrated in Fig.~\ref{fig:rotation}, we define $\zeta_0$ as the rotation angle around the $z-$axis, which yields new coordinate axes $x'$, $y'$ and $z$. Similarly, $\chi_0$ is defined as the rotation angle around the $x'-$axis. Both BS and UE are equipped with antenna arrays of arbitrary but known geometries and communicate through a mmWave channel.  

Although a device may have up to three rotation angles, we consider two angles because the estimation of three orientation angles is not possible with only LOS communication.  Hence, our formulation is representative of practical applications characterized by two rotation angles\footnote{This corresponds for instance to a vehicle that can turn left and right ($\zeta_0$) or ascend and descend ($\chi_0$), but not slip or flip.}, such as near-static\footnote{{We study positioning with a short signal snapshot,  during which the UE moves by a negligible distance. Subsequently, there has to be another layer where the snapshot positions are filtered through tracking techniques and mobility models, but this is out of the scope of this paper.}}

Our objective is to derive the performance bounds of estimating $\mathbf{p}$ and $\mathbf{o}$ via TOA, DOA, and DOD estimation, in the presence of the unknown nuisance parameters, i.e., the timing offset between the BS and UE clocks, ${B}$, and the unknown channel. This is done for the RLP and CLP protocols described in Section \ref{sec:protcols}. Our analysis considers the effect of all these unknown parameters. If a subset of the parameters is known, the bounds become lower and can be easily derived as special cases.

\subsection{Channel Model}
We consider protocols initiated by either the BS or UE. The device initiating the protocol is denoted by D$_1$, and the responding device by D$_2$. {In the presence of multipath, mmWave paths are orthogonal and information-additive \cite{Witrisal2016,Leitinger2015,Rico2018}, and hence do not interfere with one another. Moreover, the LOS path is stronger than the NLOS paths and hence provides the highest useful information in terms of positioning, while also being easy to isolate based on the signal power profile. Therefore, although we assume that the exchange of signals occurs via the LOS path, our analysis is valid even when there are NLOS paths. In any case, the presence of NLOS paths would assist localization, unlike in other systems, e.g., GPS, where multipath can limit the performance \cite{Witrisal2016,Leitinger2015,Rico2018}.}

\begin{rem}[Notation] All parameters related to D$_1$ and D$_2$ are denoted by the subscripts ``1" and ``2", respectively. Moreover, the superscripts ``f" and ``b" are used to relate the parameters to the forward and backward transmissions, respectively. Also, unless otherwise stated, all the provided times are with respect to the clock of D$_1$, which is considered a global clock. See Fig.~\ref{fig:model}.
\end{rem}

\begin{figure}[!t]
	\centering
	\includegraphics[scale=0.70]{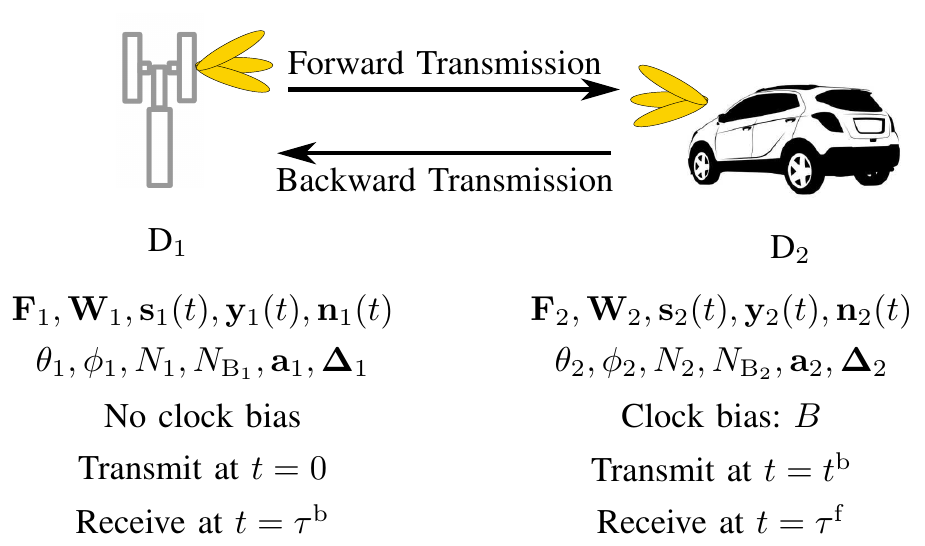}
	\caption{Summary of parameters at D$_1$ and D$_2$. Although D$_1$ and D$_2$ in the figure are BS and UE, this can be reversed.}
	\label{fig:model}
\end{figure}

Let $h^\mathrm{f}\triangleq\beta^\mathrm{f}\exp(j\psi^\mathrm{f})$ be the complex LOS path gain in the forward direction, $N_{\mathrm{1}}$ and $N_\mathrm{2}$ be the number of antennas at D$_1$ and D$_2$, respectively, and $(\theta_{\mathrm{1}},\phi_{\mathrm{1}})$ and $(\theta_{\mathrm{2}},\phi_{\mathrm{2}})$ be the forward DOD and DOA at D$_1$ and D$_2$, respectively. Also, define 
$\boldsymbol\vartheta\triangleq[\theta_{\mathrm{1}},\phi_{\mathrm{1}},\theta_{\mathrm{2}},\phi_{\mathrm{2}}]^\mathrm{T}$.

The \textit{forward} signal, from D$_1$ to D$_2$, undergoes a forward channel {given by \cite{Heath2016}}
\begin{align}
\mathbf{H}^\mathrm{f}(\boldsymbol\vartheta, \tau^\mathrm{f},h^\mathrm{f})&\triangleq\mathbf{H}^\mathrm{f}_\mathrm{s}(\boldsymbol\vartheta,h^\mathrm{f})\delta (t-{\tau^\mathrm{f}}),\in\mathbb{C}^{N_{\mathrm{2}}\times{N_\mathrm{1}}} \label{eq1}
\end{align}
where $\delta(t)$ is the Dirac delta function, $t=\tau^\mathrm{f}$ is the perceived TOA at D$_2$, and
\begin{align}
\mathbf{H}^\mathrm{f}_\mathrm{s}(\boldsymbol\vartheta,h^\mathrm{f})&\triangleq\sqrt{N_{\mathrm{1}}N_\mathrm{2}}h^\mathrm{f}\mathbf{a}_{\mathrm{2}}(\theta_{\mathrm{2}},\phi_{\mathrm{2}})\mathbf{a}^\mathrm{T}_{\mathrm{1}}(\theta_{\mathrm{1}},\phi_{\mathrm{1}}).\label{eq:channel_model_f}
\end{align}
$\mathbf{a}_{i}, i\in\{1,2\}$ is the response vectors at D$_i$ given by
\begin{align}
\mathbf{a}_{i}(\theta_{i},\phi_{i})&\triangleq\frac{1}{\sqrt{N_{i}}}e^{-j\boldsymbol{\Delta}_{i}^\mathrm{T}\mathbf{k}(\theta_{i},\phi_{i})},\qquad\in\mathbb{C}^{N_i}\label{eq:a_t}
\end{align}
where  $\mathbf{k}(\theta,\phi)=\frac{2\pi}{\lambda}[\cos\phi\sin\theta, \sin\phi\sin\theta, \cos\theta]^\mathrm{T}$ is the wavenumber vector, $\lambda$ is the wavelength,  $\boldsymbol\Delta_{i}\in\mathbb{C}^{3\times{N_i}}$ is a matrix whose columns contain the 3D Cartesian coordinates of the array elements of D$_i$ in meters. For brevity, we drop the angle parameters from the notation of $\mathbf{a}_{i}$. 

Similarly, the backward channel from D$_2$ to D$_1$ is defined as
\begin{align}
\mathbf{H}^\mathrm{b}(\boldsymbol\vartheta,\tau^\mathrm{b},h^\mathrm{b})&\triangleq\mathbf{H}^\mathrm{b}_\mathrm{s}(\boldsymbol\vartheta,h^\mathrm{b})\delta (t-{\tau^\mathrm{b}})\in\mathbb{C}^{N_{\mathrm{1}}\times{N_\mathrm{2}}},
\end{align}
where $h^\mathrm{b}\triangleq\beta^\mathrm{b}\exp(j\psi^\mathrm{b})$ and
\begin{align}
\mathbf{H}^\mathrm{b}_\mathrm{s}(\boldsymbol\vartheta,h^\mathrm{b})&\triangleq\sqrt{N_{\mathrm{1}}N_\mathrm{2}}h^\mathrm{b}\mathbf{a}_{\mathrm{1}}(\theta_{\mathrm{1}},\phi_{\mathrm{1}})\mathbf{a}^\mathrm{T}_{\mathrm{2}}(\theta_{\mathrm{2}},\phi_{\mathrm{2}}),\label{eq:channel_model_b}
\end{align}
where $\tau^\mathrm{b}$ denotes the local TOA at D$_1$.

{Note that \eqref{eq1}--\eqref{eq:channel_model_b} represent an accepted model for mmWave channels \cite{Heath2016}. Unlike cmWave channels, which experience rich scattering and relatively low propagation losses, mmWave channels are sparse and have high propagation losses, leading to weaker NLOS paths than LOS. Furthermore, due to the large temporal and spatial resolution of mmWave massive MIMO systems, reflections can be resolved if there are NLOS paths, and the parameters of the LOS can be estimated without noticeable impact from the NLOS \cite{Zohair2017}. Thus, for the sake of analysis, one can consider that the LOS-only situation is representative of scenarios where the reflections are resolvable, if present at all. For the cases where the LOS path is blocked, it has been shown recently that the probability of localization via NLOS paths alone is only about 12\% \cite{Lone2019}.}

\begin{figure*}[!t]
	\centering
	\includegraphics[scale=0.7]{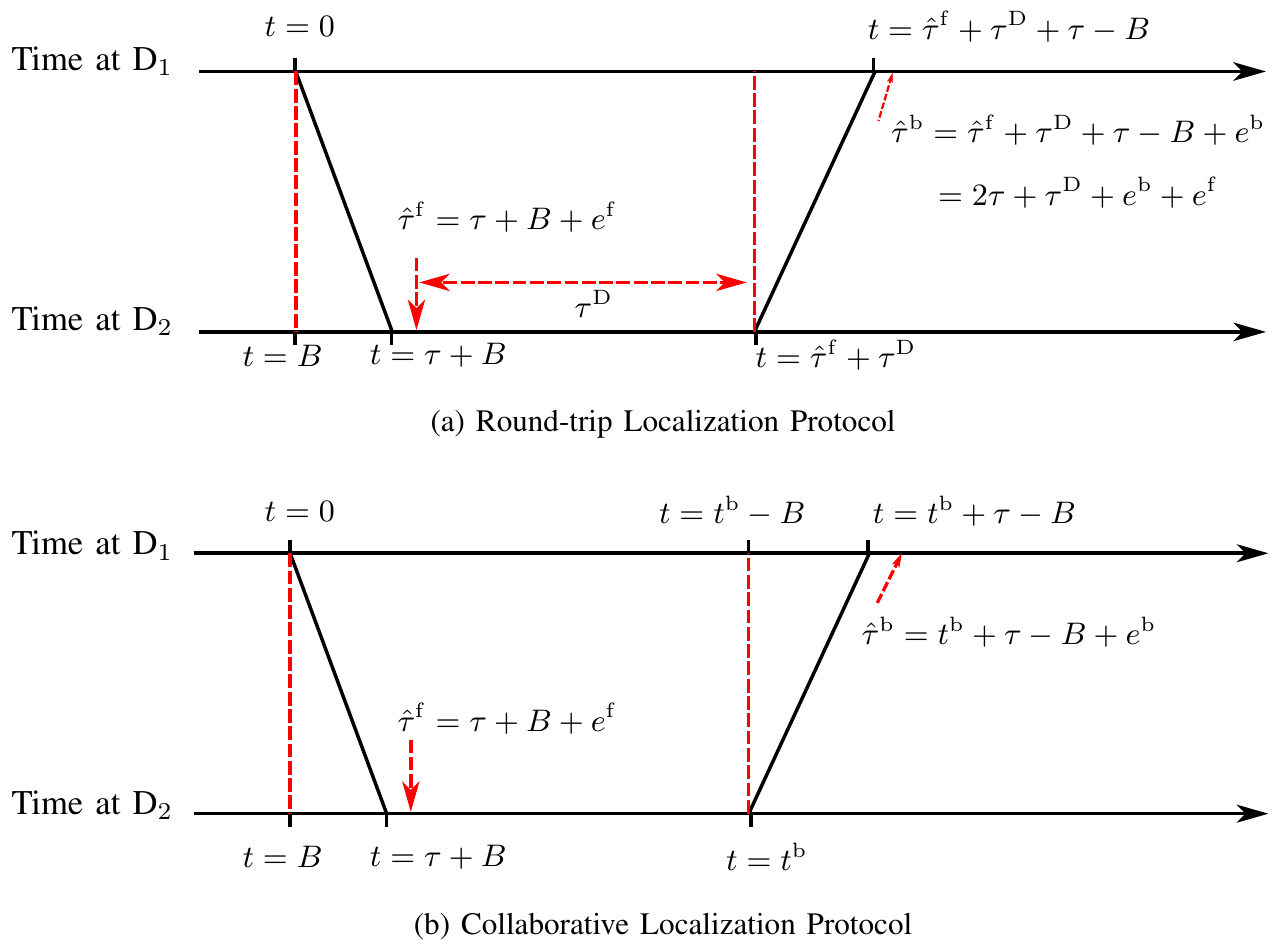}
	\caption{The timeline of the studied TWL protocols.}
	\label{fig:protocols}
\end{figure*}

\subsection{Precoding and Combining}
The signal transmitted from D$_1$ is modeled by $\sqrt{E_\mathrm{t}} \mathbf{F}_1\mathbf{s}_1(t)$, where $E_\mathrm{t}$ is the transmitted energy per symbol, and $\mathbf{F}_1\in\mathbb{C}^{N_\mathrm{1}\times{N_{\mathrm{B},1}}}$ is the transmit beamforming matrix at D$_1$ containing $N_{\mathrm{B},1}$ analog beamforming vectors. The pilot signal $\mathbf{s}_1(t)\triangleq[s_{1,1}(t),s_{1,2}(t),...,s_{1,N_{\mathrm{B}_1}}(t)]^\mathrm{T}$ is written as
\begin{align}
s_{1,b}(t)=\sum_{\ell=0}^{N_{\mathrm{s}}-1}a^{(b)}_{1,\ell}{g(t-\ell T_\mathrm{s})},\ 1\leq b\leq{N_{\mathrm{B_1}}}, \label{eq:txsignal}
\end{align}
where $a^{(b)}_{1,\ell}$ are known unit-energy pilot symbols transmitted over the $b^{\mathrm{th}}$ beam from D$_1$, and  $g(t)$ is a unit-energy pulse with a {symmetric} power spectral density (PSD), denoted by $|{G(f)}|^2$. In \eqref{eq:txsignal}, $N_{\mathrm{s}}$ is the number of pilot symbols and $T_{\mathrm{s}}$ is the symbol duration, leading to a total observation time of  $T_{\mathrm{o}} \approx N_{\mathrm{s}}T_{\mathrm{s}}$. Note that we keep the transmitted power fixed with $N_1$ by setting $\mathrm{Tr}\left(\mathbf{F}_1^\mathrm{H}\mathbf{F}_1\right)=1$, and   $\mathbf{s}_1(t)\mathbf{s}_1^\mathrm{H}(t)=\mathbf{I}_{N_{\mathrm{B_1}}}$, where $\mathrm{Tr}\left(\cdot\right)$ denotes the matrix trace, and $\mathbf{I}_{N_{\mathrm{B_1}}}$ is the $N_{\mathrm{B_1}}$-dimensional identity matrix.	Similarly, $\mathbf{W}_2\in\mathbb{C}^{N_\mathrm{2}\times{N_{\mathrm{B},2}}}$ denotes the receive beamforming matrix at D$_2$ containing $N_{\mathrm{B},2}$ analog beamforming vectors.

In backward transmission, D$_2$ transmits $\mathbf{s}_2(t)$ via a beamforming matrix, $\mathbf{F}_2$ containing $N_{\mathrm{B}_2}$ beams, while D$_1$ receives via a beamforming matrix, $\mathbf{W}_1$ containing $N_{\mathrm{B}_1}$ beams. Both $\mathbf{F}_2$ and $\mathbf{W}_1$ are defined similar to $\mathbf{W}_2$ and $\mathbf{F}_1$, respectively, but with possibly different beam directions.

\section{{Synchronization and Localization Protocols}}\label{sec:protcols}
In this section, {we discuss how clock synchronization can be addressed in 5G mmWave}. We start by presenting a general formulation of two-way localization, which we then specify for two different protocols with the aid of Fig. \ref{fig:protocols}.
\subsection{General Operation}
We take the clock at D$_1$ as a reference and assume that D$_2$ has a \textit{clock bias}\footnote{{Bias is modeled as an unknown constant, as we consider a snapshot observation over which it is assumed to remain unchanged.}}, ${B}$, with respect to it. We also denote the nominal TOA by $\tau=\|\mathbf{p}\|/c$, where $c$ is the speed of light.

During the \textbf{forward transmission}, the signal received after beamforming at D$_2$ is given by
\begin{align}
	\mathbf{y}_2(t)&=\sqrt{E_\mathrm{t}}\mathbf{W}_2^\mathrm{H}\mathbf{H}^\mathrm{f}_\mathrm{s}(h^\mathrm{f},\boldsymbol\vartheta)\mathbf{F}_1\mathbf{s}_1(t-\tau^\mathrm{f})+\mathbf{n}_2(t),\label{eq:y2}
\end{align}
where $\mathbf{n}_2(t)$ is zero-mean additive \textit{spatially-correlated} Gaussian noise, since the received signals are observed at the beamformer output. Therefore, the corresponding noise auto-covariance matrix is $\mathbf{R}_\mathrm{n2}=N_0\mathbf{W}_2^\mathrm{H}\mathbf{W}_2$, where $N_0$ is the noise PSD. We assume that $N_0$ is identical at BS and UE. Moreover, the delay at D$_2$ is 
\begin{align}
\tau^\mathrm{f}=\tau+B.\label{eq:tauf}
\end{align}

	Similarly, in the \textbf{backward transmission}, the signal received after beamforming at D$_1$ is
\begin{align}
\mathbf{y}_1(t)&=\sqrt{E_\mathrm{t}}\mathbf{W}_1^\mathrm{H}\mathbf{H}^\mathrm{b}_\mathrm{s}(\boldsymbol\vartheta,h^\mathrm{b})\mathbf{F}_2\mathbf{s}_2(t-\tau^\mathrm{b})+\mathbf{n}_1(t)\label{eq:y1},
\end{align}
where  $\mathbf{n}_1(t)$ has an auto-covariance matrix $\mathbf{R}_\mathrm{n1}=N_0\mathbf{W}_1^\mathrm{H}\mathbf{W}_1$. {Note that the backward transmission is initiated by D$_2$ at {a} time $t=t^\mathrm{b}$, and that the clock bias of D$_2$ observed at $D_1$ is $-B$}. Hence, the delay at D$_1$ is
\begin{align}
	\tau^\mathrm{b}=t^\mathrm{b}+\tau-{B}.\label{eq:taub}
\end{align} 

{There are different ways by which the synchronization of the response message from D$_2$ can be coordinated. In the following, we specify our formulation for two localization protocols, round-trip (RLP) and collaborative (CLP). While $\tau^\mathrm{f}$ is the same for CLP and RLP, their essential difference is in how each one defines $t^\mathrm{b}$, the instant at which D$_2$ sends the reply message (backward). For RLP, D$_2$ starts transmission after a pre-defined time-interval $\tau_\mathrm{D}$, taken with reference to its local clock, while in CLP, it starts transmission after $t^\mathrm{b}$, taken with reference to the clock of D$_2$.}

\subsection{Round-Trip Localization Protocol (RLP)}\label{sec:proptocols}
Under RLP, D$_2$ estimates ${\tau}^\mathrm{f}$ and waits for a pre-agreed delay $\tau^\mathrm{D}$ before transmitting back the signal $\mathbf{s}_2(t)$. In other words, 
\begin{align}\label{eq:tb}
t^\mathrm{b}=\hat{\tau}^\mathrm{f}+\tau^\mathrm{D},
\end{align} 
{where $(\hat\cdot)$ denotes the estimated value of a parameter.} See Fig.~\ref{fig:protocols}(a). We now introduce $e^\mathrm{f}\triangleq\hat{\tau}^\mathrm{f}-\tau^\mathrm{f}$ (and similarly $e^\mathrm{b}\triangleq\hat{\tau}^\mathrm{b}-\tau^\mathrm{b}$). Substituting \eqref{eq:tb} in \eqref{eq:taub}, then using \eqref{eq:tauf}, it can be shown that D$_1$ receives the signal $\mathbf{y}_1(t)$ at time
\begin{align}
{\tau}^\mathrm{b}&=\hat{\tau}^\mathrm{f}+\tau^\mathrm{D}+\tau-{B}=2\tau+e^\mathrm{f}+\tau^\mathrm{D},\label{eq:tau_b}
\end{align} 
 
Finally, based on $\mathbf{y}_1(t)$, D$_1$ estimates $\hat{\tau}^\mathrm{b}$ and eventually determines $\mathbf{p}$, and $\mathbf{o}$. Note that $B$ in the forward and backward transmissions cancel out and need not be estimated at D$_2$.

\subsection{Collaborative Localization Protocol (CLP)}	
Unlike RLP, {under CLP D$_2$ sends back a signal $\mathbf{s}_2(t)$ at a pre-agreed time instant $t=t^\mathbf{b}$. The value of $t=t^\mathbf{b}$ can be chosen to be large enough to avoid overlapping with the preceding transmission of $\mathbf{s}_1(t)$. Given that D$_2$ decides that the instant $t=t^\mathbf{b}$ has occurred based on its own clock, then the TOA measured by D$_1$ in its own time scale is given by \eqref{eq:taub}.}

In parallel, D$_1$ also receives $\mathbf{y}_2(t)$ via an error-free feedback\footnote{{To give a general exposition, we assume that the second signal is sent back entirely. However, there are some alternatives that facilitate obtaining the same bounds in a more practical way, like feeding back the parameters estimated from $\mathbf{y}_2(t)$ instead of the actual $\mathbf{y}_2(t)$.} {Any errors introduced in the transmission are assumed to be corrected via layers of coding and ARQ.}} link that can possibly be established using a microwave channel. Finally, based on $\mathbf{y}_1(t)$ \textit{and} $\mathbf{y}_2(t)$, D$_1$ estimates $\mathbf{p}$ and $\mathbf{o}$. Comparing \eqref{eq:tau_b} and \eqref{eq:taub}, it can be seen that $B$ needs to be estimated under CLP, unlike RLP. 

\section{Derivation of the Two-Way Position and Orientation Error Bounds } \label{sec:peb_oeb}
{After defining the system model and the communication protocols that govern the observations collection, we now proceed to define and derive PEB and OEB as performance metrics for the two protocols. These metrics are lower bounds on the performance of any estimator and can thus be used to benchmark localization algorithms. In fact, these bounds are tight for the problem under investigation. That is, the performance of well-designed practical algorithms approaches these bounds in the localization scenarios of interest \cite{Shahmansoori2017}. Therefore, analyzing the protocols in terms of the PEB and the OEB has the advantage of being representative of practical designs without the need for proposing detailed estimation algorithms. Moreover, since the PEB and OEB can often be computed in closed forms, another advantage is that they provide fundamental insights into the localization problem.}

{The PEB and OEB are derived from the FIM, a notion we discuss first in Section \ref{sec:basicFIM}. Then, we apply the FIM to the estimation of channel parameters in the forward and backward transmissions in Section \ref{sec:channelFIM}. This allows us to compute the PEB and OEB of RLP and CLP  in Sections \ref{sec:RLP} and \ref{sec:CLP}, respectively, and make a quantitative performance comparison in Section \ref{sec:comparison}. 
}

\subsection{Basic FIM Concepts}\label{sec:basicFIM}
{In this section, we digress to provide a brief introduction to the notion of FIM and Equivalent FIM (EFIM), useful in the analysis of the TWL protocols.} {For more background on Fisher information, the reader is referred to \cite{kay1993}.}

{Given a vector observation $\mathbf{y}$ and an unknown deterministic vector parameter $\boldsymbol{\theta}$, related by $\mathbf{y} = \mathbf{h}(\boldsymbol{\theta}) + \mathbf{n}$, where $\mathbf{n} \sim \mathcal{N}(\mathbf{0},\boldsymbol{\Sigma})$, with $\boldsymbol{\Sigma}$ independent of $\boldsymbol{\theta}$, then the FIM $\mathbf{J}_{\boldsymbol{\theta}}$ is a positive semi-definite matrix, defined as $\mathbf{J}_{\boldsymbol{\theta}} = \nabla_{\boldsymbol{\theta}}\mathbf{h}^{\mathrm{T}}(\boldsymbol{\theta})\boldsymbol{\Sigma}^{-1} (\nabla_{\boldsymbol{\theta}}\mathbf{h}(\boldsymbol{\theta}))$. Under certain regularity conditions, the inverse of the FIM (provided it exists) serves as a lower bound on the estimation error covariance of any unbiased estimator: 
	\begin{align}
	\mathbb{E}\{ (\hat{\boldsymbol{\theta}}-\boldsymbol{\theta})(\hat{\boldsymbol{\theta}}-\boldsymbol{\theta})^{\mathrm{T}}\} \succeq \mathbf{J}^{-1}_{\boldsymbol{\theta}},
	\end{align}
	where the expectation is over the noise and $\mathbf{A}\succeq \mathbf{B}$ means that $\mathbf{A}- \mathbf{B}$ is a positive semidefinite matrix. The Cram\'er-Rao lower bound (CRLB) is computed as the diagonal of the inverse FIM.  }

{If instead of $\boldsymbol{\theta}$ we need the FIM of $\boldsymbol{\phi}=f(\boldsymbol{\theta})$, we can apply a transformation on the FIM.
	\begin{definition}[FIM Transformation]\label{def:tfim}
Given the FIM $\mathbf{J}_{\boldsymbol{\theta}}$ and an injective mapping $\boldsymbol{\theta}=f(\boldsymbol{\phi})$, the FIM $\mathbf{J}_{\boldsymbol{\phi}}$ is given by  {\cite{kay1993}}
\begin{align} \label{eq:tfim}
\mathbf{J}_{\boldsymbol{\phi}} = \boldsymbol{\Upsilon}\mathbf{J}_{\boldsymbol{\theta}}\boldsymbol{\Upsilon}^{\mathrm{T}},
\end{align}
where $\boldsymbol{\Upsilon}$ is a Jacobian matrix with $[\boldsymbol{\Upsilon}]_{i,j}=\partial \boldsymbol{\theta}_i / \partial \boldsymbol{\phi}_j = \partial [f(\boldsymbol{\phi})]_i / \partial \boldsymbol{\phi}_j$.
	\end{definition}
}

{The EFIM is derived from the FIM, when we are interested only in part of the vector $\boldsymbol{\theta}$. 
	\begin{definition}[Equivalent FIM]\label{def:efim}
Given a parameter vector $\boldsymbol{\theta}\triangleq[\boldsymbol{\theta}_1^{\mathrm{T}},\boldsymbol{\theta}_2^{\mathrm{T}}]^{\mathrm{T}}$ with associated FIM
\begin{align}
\mathbf{J}_{\boldsymbol{\theta}}=\begin{bmatrix}\mathbf{J}_{\boldsymbol{\theta}_1}&\mathbf{J}_{\boldsymbol{\theta}_1\boldsymbol{\theta}_2}\\
\mathbf{J}^\mathrm{T}_{\boldsymbol{\theta}_1\boldsymbol{\theta}_2}&\mathbf{J}_{\boldsymbol{\theta}_2}
\end{bmatrix},
\end{align}
Then, the EFIM of $\boldsymbol{\theta}_1$ is given by Schur complement as  {\cite{Shen2010_2}}
\begin{align}
\mathbf{J}^{\mathrm{e}}_{\boldsymbol{\theta}_1}=\mathbf{J}_{\boldsymbol{\theta}_1}-\mathbf{J}_{\boldsymbol{\theta}_1\boldsymbol{\theta}_2}\mathbf{J}^{-1}_{\boldsymbol{\theta}_2}\mathbf{J}^\mathrm{T}_{\boldsymbol{\theta}_1\boldsymbol{\theta}_2}.
\end{align}
	\end{definition}
	Note that according to this definition, $\mathbf{J}_{\boldsymbol{\theta}_1}$ is the FIM of $\boldsymbol{\theta}_1$ if $\boldsymbol{\theta}_2$ were known, and $\mathbf{J}_{\boldsymbol{\theta}_1\boldsymbol{\theta}_2}\mathbf{J}^{-1}_{\boldsymbol{\theta}_2}\mathbf{J}^\mathrm{T}_{\boldsymbol{\theta}_1\boldsymbol{\theta}_2}$ is the loss of information due to the uncertainty of $\boldsymbol{\theta}_2$.}

\begin{definition}[PEB and OEB]
	Given the equivalent Fisher information matrix of the orientation and the position, {$\mathbf{J}_{\mathbf{o,p}}^\mathrm{e}\triangleq\mathbf{C}^{-1}\in\mathbb{R}^{5\times{5}}$}, then, the OEB and PEB are defined as \cite{Shen2010_2}
	\begin{subequations}
\begin{align}
\mathrm{OEB}&\triangleq\sqrt{{\left[\mathbf{C}\right]_{1,1}+\left[\mathbf{C}\right]_{2,2}}},\\
\mathrm{PEB}&\triangleq\sqrt{{\left[\mathbf{C}\right]_{3,3}+\left[\mathbf{C}\right]_{4,4}+\left[\mathbf{C}\right]_{5,5}}}
\end{align}
	\end{subequations}	
\end{definition}	
\subsection{General FIM for Channel Parameters}\label{sec:channelFIM}
For either the forward or the backward transmission, we can compute the FIM of the channel parameters. Focusing on the backward transmission, the FIM of  the channel parameters $\boldsymbol{\varphi}^\mathrm{b}\triangleq\left[\boldsymbol{\vartheta}^\mathrm{T},\psi^\mathrm{b}, \beta^\mathrm{b},\tau^\mathrm{b}\right]^\mathrm{T}$ from the observation $\mathbf{y}_1(t)$ is derived in Appendix \ref{sec:app_fim_dl} 
and is shown to be 
\begin{align}
\mathbf{J}_{\boldsymbol{\varphi}^\mathrm{b}}&\triangleq\begin{bmatrix}
\mathbf{J}^\mathrm{b}_{\mathrm{SS}}&\mathbf{0}_{5\times{2}}\\
\mathbf{0}_{2\times{5}}&\left[\begin{array}{cc} J_{\beta^\mathrm{b}}^\mathrm{b}&0\\0&J_\tau^\mathrm{b} \end{array}\right]
\end{bmatrix},\label{eq:partition_fim}
\end{align}
where,
\begin{align}
\mathbf{J}^\mathrm{b}_{\mathrm{SS}}=\begin{bmatrix}
\mathbf{J}_{\boldsymbol{\vartheta}}^\mathrm{b}&\mathbf{J}^\mathrm{b}_{\boldsymbol{\vartheta}\beta^\mathrm{b}}\\
\left(\mathbf{J}_{\boldsymbol{\vartheta}\beta^\mathrm{b}}^\mathrm{b}\right)^\mathrm{T}&J^\mathrm{b}_{\beta^\mathrm{b}}
\end{bmatrix},\label{eq:Jss}
\end{align}
is the FIM corresponding to the spatial parameters of $\mathbf{J}_{\boldsymbol{\varphi}^\mathrm{b}}$, such that
\begin{align}
\mathbf{J}_{\boldsymbol{\vartheta}}^\mathrm{b}&\triangleq
\begin{bmatrix}
J^\mathrm{b}_{\theta_1}&J^\mathrm{b}_{\theta_1\phi_1}&J^\mathrm{b}_{\theta_1\theta_2}&J^\mathrm{b}_{\theta_1\phi_2}\\
J^\mathrm{b}_{\theta_1\phi_1}&J^\mathrm{b}_{\phi_1}&J^\mathrm{b}_{\phi_1\theta_2}&J^\mathrm{b}_{\phi_1\phi_2}\\
J^\mathrm{b}_{\theta_1\theta_2}&J^\mathrm{b}_{\phi_1\theta_2}&J^\mathrm{b}_{\theta_2}&J^\mathrm{b}_{\theta_2\phi_2}\\
J^\mathrm{b}_{\theta_1\phi_2}&J^\mathrm{b}_{\phi_1\phi_2}&J^\mathrm{b}_{\theta_2\phi_2}&J^\mathrm{b}_{\phi_2}\\
\end{bmatrix},\label{eq:jthetatheta}\\
\left(\mathbf{J}_{\boldsymbol{\vartheta}\beta^\mathrm{b}}^\mathrm{b}\right)^\mathrm{T}&\triangleq
\begin{bmatrix}
J^\mathrm{b}_{\theta_1\beta^\mathrm{b}}&J^\mathrm{b}_{\phi_1\beta^\mathrm{b}}&J^\mathrm{b}_{\theta_2\beta^\mathrm{b}}&J^\mathrm{b}_{\phi_2\beta^\mathrm{b}}
\end{bmatrix}.
\end{align}

The FIM of $\boldsymbol{\varphi}^\mathrm{f}\triangleq\left[{\boldsymbol{\vartheta}}^\mathrm{T},\psi^\mathrm{f}, \beta^\mathrm{f},\tau^\mathrm{f}\right]^\mathrm{T}$ is obtained from the observation $\mathbf{y}_2(t)$ in the same way and exhibits the same structure, as highlighted at the end of Appendix \ref{app:proof_theo}.
{
\begin{rem}Due to the structure of \eqref{eq:partition_fim}, the delay is always independent of the other channel parameters and can thus be treated separately. It will be convenient to introduce the EFIM of the delay in forward and backward transmissions: we denote by $J_{\tau^\mathrm{f}}$ the EFIM of $\tau^\mathrm{f}$, obtained from applying Definition \ref{def:efim} to the FIM of $\left[{\boldsymbol{\vartheta}}^\mathrm{T},\psi^\mathrm{f}, \beta^\mathrm{f},\tau^\mathrm{f} \right]^\mathrm{T}$ based on the measurement $\mathbf{y}_2(t)$. Similarly, 
	we denote by $J_{\tau^\mathrm{b}}$ the EFIM of $\tau^\mathrm{b}$, obtained from applying Definition \ref{def:efim} to the FIM of $\left[\boldsymbol\vartheta^\mathrm{T},\psi^\mathrm{b}, \beta^\mathrm{b},\tau^\mathrm{b} \right]^\mathrm{T}$ based on the measurement $\mathbf{y}_1(t)$. Note that, by definition, 
	\begin{align}
	\quad\mathbb{E}\{\left(e^\mathrm{f}\right)^2\}\geq{J_{\tau^\mathrm{f}}^{-1}},\qquad &\mathbb{E}\{\left(e^\mathrm{b}\right)^2\}\geq{J_{\tau^\mathrm{b}}^{-1}}.\label{eq:fim_eb}
	\end{align}
\end{rem}
	}
\subsection{PEB and OEB for RLP}\label{sec:RLP}
To compute PEB and OEB, we first need the EFIM of the position and orientation, $\mathbf{J}_{\mathbf{o,p}}^{\mathrm{e,b}}$. However, $\mathbf{p}$ and $\mathbf{o}$ are functions of $\boldsymbol{\vartheta}$ and $\tau$ and, hence $\mathbf{J}_{\mathbf{o,p}}^{\mathrm{e,b}}$ can be obtained as a transformation of the EFIM of $\boldsymbol{\vartheta}$ and $\tau$ as outlined in Definition 1. Since the temporal and spatial parts in \eqref{eq:partition_fim} are independent, the EFIM of $\boldsymbol{\vartheta}$ and $\tau$ is given by
\begin{align}
\mathbf{J}_{\boldsymbol{\vartheta}\tau}^\mathrm{e,b}	=\begin{bmatrix}
\mathbf{J}_{\boldsymbol{\vartheta}}^\mathrm{e,b}&\mathbf{0}_4\\
\mathbf{0}_4^\mathrm{T}&J_{\tau}
\end{bmatrix}.\label{eq:schur}
\end{align}
We now outline how to obtain $\mathbf{J}_{\boldsymbol{\vartheta}\tau}^\mathrm{e,b}$ before transforming it to obtain $\mathbf{J}_{\mathbf{o,p}}^{\mathrm{e,b}}$.

It is straight-forward from \eqref{eq:Jss} that the EFIM of $\boldsymbol{\vartheta}$ based on the backward transmission is obtained using Schur complements as
\begin{align}
\mathbf{J}_{\boldsymbol{\vartheta}}^\mathrm{e,b}=\mathbf{J}_{\boldsymbol{\vartheta}}^\mathrm{b}-\frac{1}{{J}^\mathrm{b}_{\beta^\mathrm{b}}}\mathbf{J}_{\boldsymbol{\vartheta}\beta^\mathrm{b}}^\mathrm{b}\left(\mathbf{J}_{\boldsymbol{\vartheta}\beta^\mathrm{b}}^\mathrm{b}\right)^\mathrm{T}.\label{eq:fim_angles_e}
\end{align}

According to  \eqref{eq:tau_b}, $\tau$ depends on the estimate of $\tau^\mathrm{f}$ as well as the value of $\tau^\mathrm{b}$. While we can determine $J_{\tau^\mathrm{f}}$ based on $\mathbf{y}_2(t)$,  $J_{\tau^\mathrm{b}}$ is based on $\mathbf{y}_1(t)$. Therefore, to obtain the FIM of $\tau$ rather than $\tau^\mathrm{b}$ or $\tau^\mathrm{f}$, we apply the fact that, under RLP, the delays are not dependent on any of the other channel parameters. Towards that, recall that $\hat\tau^\mathrm{b}=2\tau+e^\mathrm{f}+e^\mathrm{b}+\tau^\mathrm{D},$ and define
\begin{align}
\tau'\triangleq\frac{\hat\tau^\mathrm{b}-\tau^\mathrm{D}}{2}=\tau+\frac{e^\mathrm{f}+e^\mathrm{b}}{2}.
\end{align}
Consequently, using \eqref{eq:fim_eb} yields
\begin{align}
\mathbb{E}\Big\{\left(\tau'-\tau\right)^2\Big\}\geq \frac{1}{4}\left(J_{\tau^\mathrm{f}}^{-1}+J_{\tau^\mathrm{b}}^{-1}\right),
\end{align}
that is,
\begin{align}
{J_{\tau}=4\left(J_{\tau^\mathrm{f}}^{-1}+J_{\tau^\mathrm{b}}^{-1}\right)^{-1}}.\label{eq:jatua_dlp}
\end{align}
Note that in this scenario, the estimation of $\tau$ is less accurate than the estimation of $\tau^\mathrm{b}$ due to its further dependence of $\tau^\mathrm{f}$.

Applying Definition \ref{def:tfim}, to \eqref{eq:schur}, it can be shown that
\begin{align}
\mathbf{J}_{\mathbf{o,p}}^{\mathrm{e,b}}{|_\mathrm{RLP}}=\boldsymbol\Upsilon\mathbf{J}_{\boldsymbol{\vartheta}\tau}^\mathrm{e,b}\boldsymbol\Upsilon^{\mathrm{T}},
\end{align}
where 
\begin{align}\label{eq:transformation_matrix}
\boldsymbol{\Upsilon}&\triangleq\left[\begin{array}{cccc;{2pt/2pt}c}
\frac{\partial\theta_1}{\partial\mathbf{o}}&\frac{\partial\phi_1}{\partial\mathbf{o}}&\frac{\partial\theta_2}{\partial\mathbf{o}}&\frac{\partial\phi_2}{\partial\mathbf{o}}&\frac{\partial\tau}{\partial\mathbf{o}}\\
\frac{\partial\theta_1}{\partial\mathbf{p}}&\frac{\partial\phi_1}{\partial\mathbf{p}}&\frac{\partial\theta_2}{\partial\mathbf{p}}&\frac{\partial\phi_2}{\partial\mathbf{p}}&\frac{\partial\tau}{\partial\mathbf{p}}\\
\end{array}\right]=\left[\begin{array}{c;{2pt/2pt}c}
\boldsymbol{\Upsilon}_\mathrm{s}&\boldsymbol{\Upsilon}_\tau\end{array}\right].
\end{align}
Consequently, for RLP, we can isolate the spatial and temporal parts and write,
\begin{align}\label{eq:FIM_OP_RLP}
\mathbf{J}_{\mathbf{o,p}}^{\mathrm{e,b}}{|_\mathrm{RLP}}&=\underbrace{\boldsymbol{\Upsilon}_\mathrm{s}\mathbf{J}_{\boldsymbol{\vartheta}}^\mathrm{e,b}\boldsymbol{\Upsilon}_\mathrm{s}}_\text{Spatial Part}+\underbrace{{J}_{\tau}\boldsymbol{\Upsilon}_\tau\boldsymbol{\Upsilon}_\tau^\mathrm{T}}_\text{Temporal Part}.
\end{align}
{The entries of $\boldsymbol{\Upsilon}_\tau$ and $\boldsymbol{\Upsilon}_\mathrm{s}^{\mathrm{b}}$ are easily obtained from the relations mapping from location parameters to channel parameters and can be found in \cite{Zohair2017}, where it was concluded that $\boldsymbol{\Upsilon}_\tau$ is identical for the uplink and downlink, while $\boldsymbol{\Upsilon}_\mathrm{s}^{\mathrm{b}}$ is not. This results in an asymmetry in the spatial part of the FIM. To understand the implication of this asymmetry, we note that the UE position in the uplink is a function of the DOA and TOA, while in the downlink, it is a function of the DOD and TOA. However, DOD and DOA have different CRLBs, which means that the RLP localization performance in \eqref{eq:FIM_OP_RLP} depends on whether the localization is executed in the uplink (at BS) or downlink (at UE).}

\subsection{PEB and OEB for CLP}\label{sec:CLP}
As can be inferred from \eqref{eq:taub}, we have to retrieve ${B}$ in CLP, in contrast to the RLP case. Therefore, we the vector of unknown parameters is
\begin{align}
\boldsymbol{\varphi}_\mathrm{C}\triangleq\left[\boldsymbol{\vartheta}^\mathrm{T},\psi^\mathrm{b},\beta^\mathrm{b},\psi^\mathrm{f},\beta^\mathrm{f},\tau,{B}\right]^\mathrm{T}.
\end{align}
To simplify the derivation, we treat the temporal parameters, ($\tau$ and $B$), in isolation of the spatial parameters ($\boldsymbol{\vartheta}^\mathrm{T},\psi^\mathrm{b},\beta^\mathrm{b},\psi^\mathrm{f},\beta^\mathrm{f}$) because both sets  are independent. Similar to \eqref{eq:schur}, we seek to compute
\begin{align}\label{eq:efim_clp}
\mathbf{J}^\mathrm{e}_{\boldsymbol\vartheta\tau}\triangleq\begin{bmatrix}
\mathbf{J}^\mathrm{e}_{\boldsymbol\vartheta}& \mathbf{0}_4\\
\mathbf{0}_4^\mathrm{T}&J^\mathrm{e}_\tau
\end{bmatrix},
\end{align}
where $\mathbf{J}^\mathrm{e}_{\boldsymbol\vartheta}$ is to be computed from FIM of the spatial parameters  and $J^\mathrm{e}_\tau$ is to be computed from the temporal ones.

{Since D$_2$ transmission time is independent of the TOA of $\mathbf{y}_2(t)$, the two transmissions occur in non-overlapping time slots, and the noise is independent at both sides, the forward and backward transmissions can be considered independent. Therefore, the FIMs can be added according to the following Theorem.} 
\begin{theorem}
{Consider a random process to observe the unknown parameter $\mathbf{x}$ along with the unknown nuisance parameter $\mathbf{z}_1$. Consider also another random process to observe $\mathbf{x}$ along with the unknown nuisance parameter $\mathbf{z}_2$. If both processes are independent and $\mathbf{z}_1$ and $\mathbf{z}_2$, are independent, then total EFIM of $\mathbf{x}$ is
\begin{align}
    \mathbf{J}^\mathrm{e}_\mathbf{x}=\mathbf{J}^\mathrm{e,1}_\mathbf{x}+\mathbf{J}^\mathrm{e,2}_\mathbf{x},
\end{align}
    where $\mathbf{J}^\mathrm{e,i}_\mathbf{x}$ is the EFIM of $\mathbf{x}$ obtained from the $i$-th process.}
\end{theorem}
\begin{proof}
    See Appendix \ref{app:proof_theo}.
\end{proof}
In other words, the EFIM of $\boldsymbol{\vartheta}$ can be written as $\mathbf{J}^\mathrm{e}_{\boldsymbol\vartheta}=\mathbf{J}^\mathrm{e,b}_{\boldsymbol\vartheta}+\mathbf{J}^\mathrm{e,f}_{\boldsymbol\vartheta}$ by summing the EFIMs of $\boldsymbol\vartheta$ computed from $\mathbf{y}_1(t)$ and $\mathbf{y}_2(t)$. From \eqref{eq:Jss}, it follows that
\begin{align*}
\mathbf{J}^\mathrm{e}_{\boldsymbol\vartheta}=\mathbf{J}_{\boldsymbol{\vartheta}}^\mathrm{b}+\mathbf{J}_{\boldsymbol{\vartheta}}^\mathrm{f}-\frac{1}{{J}^\mathrm{b}_{\beta^\mathrm{b}}}\mathbf{J}_{\boldsymbol{\vartheta}\beta^\mathrm{b}}^\mathrm{b}\left(\mathbf{J}_{\boldsymbol{\vartheta}\beta^\mathrm{b}}^\mathrm{b}\right)^\mathrm{T}-\frac{1}{{J}^\mathrm{f}_{\beta^\mathrm{f}}}\mathbf{J}_{\boldsymbol{\vartheta}\beta^\mathrm{f}}^\mathrm{f}\left(\mathbf{J}_{\boldsymbol{\vartheta}\beta^\mathrm{f}}^\mathrm{f}\right)^\mathrm{T}.
\end{align*}

Moreover, $J^\mathrm{e}_\tau$ can be obtained noting that in the backward transmission $\tau^\mathrm{b} = t^\mathrm{b} + \tau -B$, while in the forward transmission $\tau^\mathrm{f} = \tau + B$, and that $\tau$ is independent of any other parameters. Hence, using the transformation of parameters and the fact that the two transmissions are independent, we can write the FIM of $\left[\tau,B\right]^\mathrm{T}$ as
\begin{align}
\mathbf{J}_{\tau B}=J_{{\tau}^\mathrm{b}}\left[\begin{array}{cc} 1&-1\\-1&1 \end{array}\right]+J_{{\tau}^\mathrm{f}} \left[\begin{array}{cc} 1&1\\1&1 \end{array}\right]
\end{align}
from which EFIM of $\tau$ is obtained by Schur complement as
\begin{align}
J^{\mathrm{e}}_{\tau}&=4\left(J^{-1}_{{\tau}^\mathrm{f}}+J^{-1}_{{\tau}^\mathrm{b}}\right)^{-1}.
\label{eq:fim_tau_e}
\end{align}
{It is interesting to see that the temporal information represented by $J_{\tau}$  is identical for both CLP \eqref{eq:fim_tau_e} and RLP \eqref{eq:jatua_dlp}. Therefore, any performance difference between these two protocols is attributed to the spatial information only. }

We now derive the EFIM of the position and orientation. Based on \eqref{eq:efim_clp} and Definition \ref{def:efim}
\begin{align}
\mathbf{J}_{\mathbf{o,p}}^{\mathrm{e}}{|_\mathrm{CLP}}&=\boldsymbol\Upsilon\mathbf{J}_{\boldsymbol{\vartheta}\tau}^\mathrm{e}\boldsymbol\Upsilon^{\mathrm{T}},\notag\\
&=\underbrace{\boldsymbol{\Upsilon}_\mathrm{s}\mathbf{J}_{\boldsymbol{\vartheta}}^\mathrm{e,f}\boldsymbol{\Upsilon}_{\mathrm{s}}^\mathrm{T}}_\text{Forward Spatial Part}+\underbrace{\boldsymbol{\Upsilon}_\mathrm{s}\mathbf{J}_{\boldsymbol{\vartheta}}^\mathrm{e,b}\boldsymbol{\Upsilon}_{\mathrm{s}}^\mathrm{T}}_\text{Backward Spatial Part}+\underbrace{{J}_{\tau}\boldsymbol{\Upsilon}_\tau\boldsymbol{\Upsilon}_\tau^\mathrm{T}}_\text{Temporal Part}.\label{eq:FIM_OP_CLP}
\end{align}

Note that \eqref{eq:FIM_OP_CLP} comprises three terms: two terms related to the spatial information in the forward and backward transmissions, and one term related to the temporal information.

\subsection{Performance Comparison of RLP, CLP and OWL}\label{sec:comparison}
{The EFIM of position and orientation under RLP is given in \eqref{eq:FIM_OP_RLP}, while that under CLP is given in \eqref{eq:FIM_OP_CLP}. In this section, we compare the performance of these two protocols with the standard one-way localization (OWL) from \cite{Zohair2017}, where it was shown that 
	\begin{align}\label{eq:FIM_OP_OWL}
	{\mathbf{J}^{\mathrm{e},\mathrm{b}}_\mathbf{o,p}|_\mathrm{OWL}=\underbrace{\boldsymbol{\Upsilon}_\mathrm{s}\mathbf{J}_{\boldsymbol{\vartheta}}^\mathrm{e,b}\boldsymbol{\Upsilon}_{\mathrm{s}}^\mathrm{T}}_{\text{Spatial Part}}+\underbrace{{J}_{\tau^\mathrm{b}}\boldsymbol{\Upsilon}_\tau\boldsymbol{\Upsilon}_\tau^\mathrm{T}}_{\text{Temporal Part}},}
	\end{align}
	under the assumption of perfect synchronization between the two devices (i.e., $B=0$). 
}
\subsubsection{RLP vs. CLP} 
{Comparing RLP to CLP, we note that \eqref{eq:FIM_OP_RLP} contains only one spatial information term, related to the backward transmission, and another temporal information term. These two terms are equal to their counterparts in \eqref{eq:FIM_OP_CLP}. Hence, $	\mathbf{J}_{\mathbf{o,p}}^\mathrm{e}|_\mathrm{CLP} \succ \mathbf{J}_{\mathbf{o,p}}^{\mathrm{e,b}}{|_\mathrm{RLP}}$, meaning that CLP will always outperform RLP. Nevertheless, CLP requires additional overhead, as it involves sending back the waveform $\mathbf{y}_2(t)$ to D$_1$ (or estimated parameters with uncertainty) and thus requires an additional data transmission.}

\subsubsection{RLP vs. OWL} Inspecting \eqref{eq:FIM_OP_OWL}, it can be seen that $\mathbf{J}^{\mathrm{e},\mathrm{b}}_\mathbf{o,p}|_\mathrm{OWL}$ has the same expression as \eqref{eq:FIM_OP_RLP}, but with
\begin{align}
{J_{\tau}=J_{\tau^\mathrm{b}}.}\label{eq:fimOWL}
\end{align}
{This means the both RLP and OWL have the same spatial information but differ in the temporal information. However, it is not clear which protocol is superior. Therefore,} we provide the following proposition. 
\begin{proposition}\label{prop:1}
	RLP outperforms OWL if $	J_{\tau^\mathrm{f}}>\frac{1}{3}J_{\tau^\mathrm{b}}$.
\end{proposition}
\begin{proof}
	Comparing RLP with OWL, it can be seen that they have equal spatial, but different temporal information. Comparing \eqref{eq:jatua_dlp} with \eqref{eq:fimOWL}, for RLP to outperform OWL, we should have
	\begin{align*}
	J_{\tau^\mathrm{b}}&<4\left(J_{\tau^\mathrm{f}}^{-1}+J_{\tau^\mathrm{b}}^{-1}\right)^{-1}=J_{\tau^\mathrm{b}}\frac{4J_{\tau^\mathrm{f}}}{J_{\tau^\mathrm{f}}+J_{\tau^\mathrm{b}}},
	\end{align*} 
	which leads to	$J_{\tau^\mathrm{f}}>\frac{1}{3}J_{\tau^\mathrm{b}}$.
\end{proof}
This means that, when the bandwidth is equal in both directions, the forward link should have at least one third the signal-to-noise ratio (SNR) of the backward link for RLP to outperform OWL. From 
\eqref{eq:tau_b_tau_f}, it can be seen that this mainly depends on the transmit and receive beamforming. However, under the general case of non-identical bandwidth allocation, \eqref{eq:tau_b_tau_f} can be used to determine the values of bandwidth and SNR that satisfy the condition in Proposition 1.

\subsection{{Relationship of RLP and CLP with Channel Parameters}}\label{sec:RLP_CLP_channel_parameters}
{Since the derived position and orientation bounds are obtained through transforming the channel parameter bounds, the localization performance is ultimately affected by the channel parameter estimation accuracy. In \cite{Zohair2017}, it was concluded that the squared PEB is the sum of the CRLBs of TOA and the BS angle (DOA in the uplink or DOD in the downlink). It was also concluded that the CRLB of DOA is better than the CRLB of DOD. Extending these results to the RLP and CLP in this paper, it can be seen that from \eqref{eq:FIM_OP_RLP}, the RLP performance is governed by the backward transmission from D$_2$ to D$_1$. In other words, if the backward transmission is uplink, D$_1$ is a BS, whose angle is a DOA, which leads to a better PEB. Note that in such a case, DOD estimation error does not affect the localization performance. The opposite is true if D$_1$ is a UE.} 
	
For CLP, it can be seen that from \eqref{eq:FIM_OP_CLP} that the squared PEB is the sum of the CRLB of DOA, DOD and TOA, meaning that regardless of the BS and UE roles, the PEB is affected by the error of estimating all of these three parameters.

{The squared OEB is the sum of the CRLBs of DOA and DOD \cite{Zohair2017}, and hence it is not affected by the accuracy of the TOA estimation.}

\subsection{{Insights on NLOS}}
When the signal propagation occurs in mixed propagation environment (LOS and NLOS), the delay of the LOS path, being the first and strongest path, can still be separated and identified, while the delays of NLOS paths must be subsequently estimated. For RLP, the NLOS paths in the backward transmission can assist the positioning, as shown in \cite{Zohair2017}. On the other hand, for CLP, the parameters of the NLOS paths in the forward and backward transmissions can be estimated separately. However, this will give rise to a path association problem, whereby the set of parameters estimated from the forward transmission have to be paired with their counter parts in the backward transmission and with the scatterers or reflectors in the environment to re-establish the different paths. Moreover, when the LOS is obstructed, the localization performance is severely degraded \cite{Lone2019}.
\section{Simulation Results and Discussion}\label{sec:sims}
\subsection{Simulation Environment}
\subsubsection{System Layout and Channel} In our simulations, we investigate and compare the RLP and CLP using the position and orientation error bounds to quantify the estimation accuracy. Since both protocols involve forward and backward transmission, we selected an equal number of antennas at both the BS and the UE to make the comparison of these protocols fair\footnote{{It is understood that BSs can accommodate more complexity, and its array can have a larger size, such as that in \cite{Fan2017} where 10,000 antennas are used. However, we use an equal number of antennas on D$_1$ and D$_2$ in order to explore the intrinsic differences between the protocols. These differences result from the lack of symmetry between the two links because the orientation is only known for one device but not the other. Thus, we prevent masking our protocol analysis by using the same number of antennas.}}. Towards that, we consider a BS and a UE both with $12\times{12}$ uniform rectangular antenna array (URA) communicating via a LOS. Moreover, we assume that the BS array is located in the $xz$-plane centered about the origin $[0,0,0]^\mathrm{T}$, thus has orientation angles of $[0^\circ,0^\circ]^\mathrm{T}$. On the other hand, the UE moves freely within a diamond-shape $120^\circ$ defined by the vertices $\{(0,0,-10),(25\sqrt{3},25,-10),(0,50,-10),(-25\sqrt{3},25,-10)\}$. That is, the BS height is 10 meters. We focus on two cases of orientation angles with respect to the $z$-axis and $x$-axis: $\mathbf{o}=[\chi_0,\zeta_0]=[0^\circ,0^\circ]^\mathrm{T}$ and $\mathbf{o}=[30^\circ,30^\circ]^\mathrm{T}$ as specified in the context. Finally, at a distance $d_1$, the channel gain is modeled as $\beta^\mathrm{b}=\beta^\mathrm{f}=\frac{\lambda}{4{\pi}d_1}$ \cite{Goldsmith2005}.
\subsubsection{Transmit-Receive Model} We select the mmWave frequency of $f=38~$GHz, and bandwidth\footnote{{At these frequency and bandwidth values, beam squint is negligible. From \cite{beqmsquint}, pointing error due to beam squint is proportional to $\left(1+f/W\right)^{-1}\approx1$.}} $W=125~$MHz. We assume an ideal $\mathrm{sinc}$ pulse-shaping filter such that $W^2_\mathrm{eff}=W^2/3$. The transmitted power $E_\mathrm{t}/T_\mathrm{s}=0~$dBm, and $N_0=-170$ dBm/Hz. Furthermore, we specify the number of pilots to be $N_\mathrm{s}=~64$ pilot symbols. This yields a location-dependent SNR of
\begin{align}
\mathrm{SNR}~[\mathrm{dB}] =150.26+20\log_{10}\left(\beta^\mathrm{b}\|\mathbf{a}^\mathrm{T}_{i}\mathbf{F}^\mathrm{H}_{i}\|\|\mathbf{W}^\mathrm{H}_j\mathbf{a}_j\|\right),
\end{align}
where $i,j\in\{1,2\}, i\neq{j}$, specified depending on the communication direction being forward or backward. {Note that this SNR results from beamforming gain of all the $N_\mathrm{B_1}$ and $N_\mathrm{B_2}$ beams combined.} {Although our formulation holds for any type of analog beamforming, in our numerical simulations,} we adopt \textit{fixed} directional beamforming  {as an example scheme} similar to \cite{Zohair2017}. We consider $ 1\leq b\leq{N_{\mathrm{B}_1}}=N_{\mathrm{B}_2}=25$ beams at both the UE and BS, such that 
\begin{figure}[!t]
	\centering
	\includegraphics[scale=0.60]{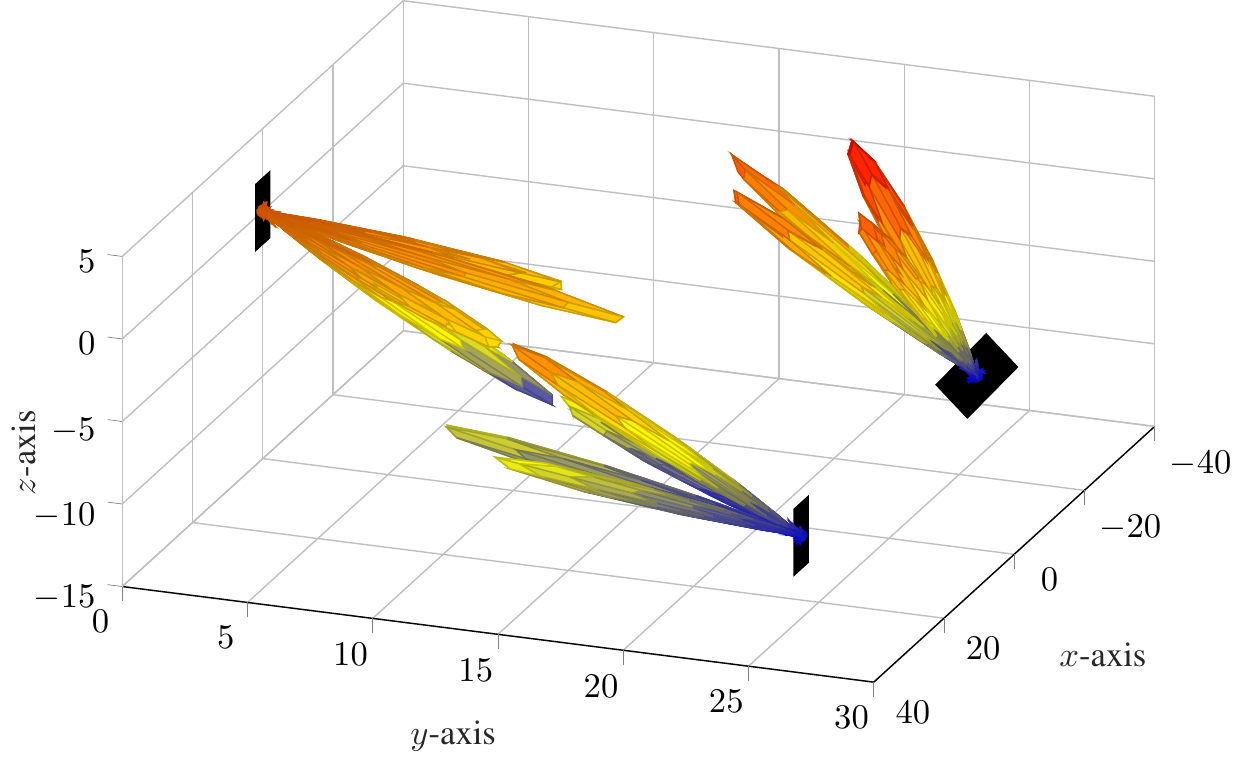}

	\caption{Beamforming example with 4 beams. The rightmost device has orientation angles of $30^\circ$, while the other two have $0^\circ$. }
	\label{fig:rotated_beams}
\end{figure}
\begin{align*}
\mathbf{f}_{1,b}&=\frac{1}{\sqrt{N_{\mathrm{B}_1}}}\mathbf{a}_{\mathrm{1}}(\theta^\mathrm{f}_{1,b},\phi^\mathrm{f}_{1,b}),\\
\mathbf{w}_{1,b}&=\frac{1}{\sqrt{N_{\mathrm{B}_1}}}\mathbf{a}_{\mathrm{1}}(\theta^\mathrm{w}_{1,b},\phi^\mathrm{w}_{1,b}),
\end{align*}
are D$_1$ transmit and receive beams pointing towards $(\theta^\mathrm{f}_{1,b},\phi^\mathrm{f}_{1,b})$ and $(\theta^\mathrm{w}_{1,b},\phi^\mathrm{w}_{1,b})$, respectively. The directions of the beams at the BS are chosen to be equispaced on the sector. On the UE, these directions are reversed to point upwards and rotated with respect to the UE frame of reference by the same orientation angles specified in the studied experiment.  This setting provides 90\% of the locations with an SNR of at least 17 dB. Fig.~\ref{fig:rotated_beams} provides three examples on beamforming configuration: a BS at $(0,0,0)$, with beams pointing downwards, a UE at $(25,25,-10)$ with zero orientation angles, and another UE at $(-25,25,-10)$ with $\mathbf{o}=[30^\circ,30^\circ]^\mathrm{T}$. The black rectangles denote the array frame of reference of the device. Note that the first UE has reversed beam direction compared to BS, while the second UE has beam directions reversed and rotated by $[30^\circ,30^\circ]^\mathrm{T}$ , so that the beam directions remain constant with respect to the UE local frame of reference.
\subsubsection{Scenarios Studied} We study the PEB and OEB under RLP and CLP and compare these bounds to those obtained for OWL in \cite{Zohair2017}. Each of these three protocols is studied when localization is performed in the uplink (at BS) and in the downlink (at UE). Recall that CLP is symmetric in both cases, hence only one curve is given.
\subsection{PEB and OEB with $0^\circ$ UE Orientation}
\begin{figure}[!t]
	\centering
	\includegraphics[scale=0.60]{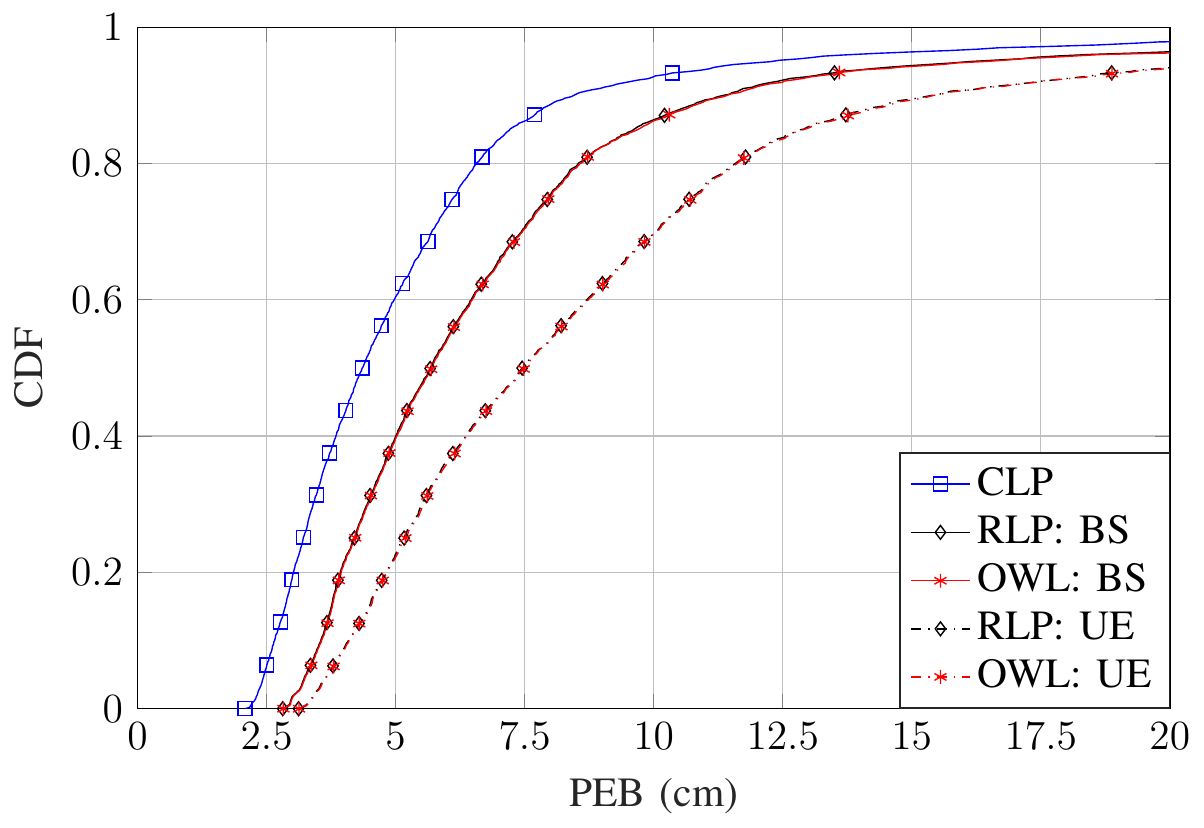}
	\caption{CDF of PEB with UE orientation angles of $0^\circ$, and $N_\mathrm{UE}=N_\mathrm{BS}=144$, $N_\mathrm{B}=25$.}
	\label{fig:peb0}
\end{figure}
\begin{figure}[!t]
	\centering
	\includegraphics[scale=0.60]{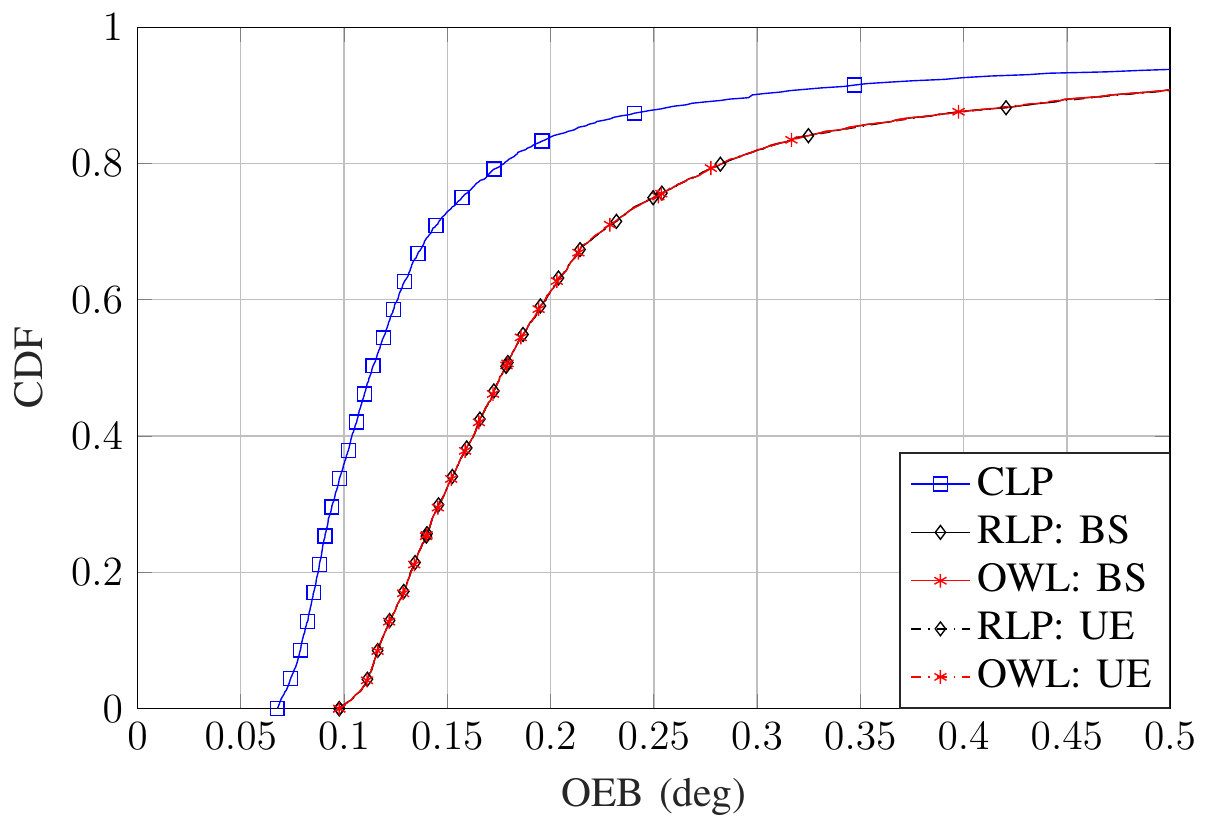}
	\caption{CDF of OEB with UE orientation angles of $0^\circ$, and $N_\mathrm{UE}=N_\mathrm{BS}=144$, $N_\mathrm{B}=25$.}
	\label{fig:oeb0}
\end{figure}

The {cumulative distribution function} (CDF) of the PEB with zero orientation angles is provided in Fig.~\ref{fig:peb0} for all the considered protocols. First of all, to have a fair comparison, we compare the three solid curves corresponding to uplink localization, and then compare those related to downlink localization (dash-dot lines). It can be seen that RLP provides a negligible improvement over OWL. Despite that, RLP is still a better approach since it alleviates the need for high-accuracy synchronization, with the cost of UE-BS coordination. As discussed in Section \ref{sec:comparison}, RLP and OWL have the same spatial component, but RLP has higher temporal information content. However, Fig.~\ref{fig:peb0} shows almost identical results for both protocols, which means that the additional temporal information in RLP is of little importance, and thus the localization performance is limited by the angles estimation rather than the time delay. To understand this phenomenon more, we study the impact of the bandwidth on the performance later in Section \ref{sec:W}. On the other hand, as expected, CLP represents the best approach among the three studied, since it attains more useful information. However, this requires a more complex implementation due to the need for a feedback channel.

{Comparing the dash-dotted curves with the solid curves in Fig.~\ref{fig:peb0}, it can be seen that the three protocols behave in the downlink in a manner similar to the uplink. And, it can also be seen that while OWL and RLP are almost identical, CLP is superior to both. However, the reasons why the performance of RLP and OWL is worse in the downlink are beyond the scope of this paper and were extensively studied in \cite{Zohair2017}. Briefly,} it was concluded that, under matched orientation between the BS and UE, the uplink PEB is better than the downlink PEB. This is because 1) PEB is a function of the CRLB of the BS angles, and 2) CRLB of DOA is lower than CRLB of DOD. Therefore, when the BS angles are DOAs (uplink), the PEB will be lower.

Considering the CDF of the  OEB with zero orientation angles in Fig.~\ref{fig:oeb0}, it can be seen that RLP and OWL exhibit identical performance. Note that OEB depends on DOA and DOD, while the enhancement of RLP over OWL is in the temporal domain. Furthermore, in line with the results in \cite{Zohair2017} with zero orientation angles, the uplink and downlink OEB are the same. Therefore, the four curves of RLP and OWL with uplink and downlink localization coincide. Moreover, in terms of OEB, CLP is also better than RLP and OWL due to the fourth term in \eqref{eq:FIM_OP_CLP}, which accounts for the coupling between the path gain and the transmission angles, providing more spatial information on the orientation angles. {Intuitively, this higher information is a result of estimating the path gain in both transmissions.}
\subsection{PEB and OEB with $30^\circ$ UE Orientation}\label{sec:sim_peb_oeb}
\begin{figure}[!t]
	\centering
	\includegraphics[scale=0.60]{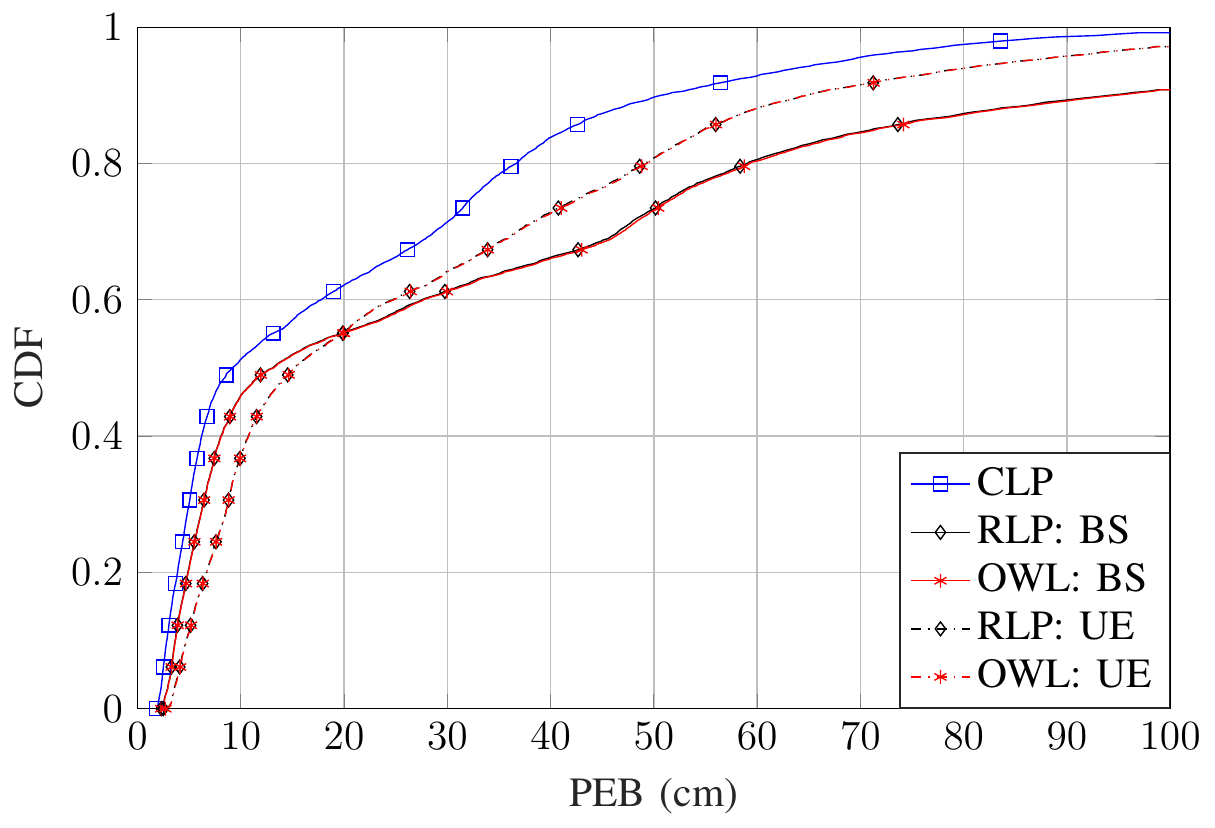}
	\caption{CDF of PEB with UE orientation angles of $30^\circ$, and $N_\mathrm{UE}=N_\mathrm{BS}=144$, $N_\mathrm{B}=25$.}
	\label{fig:peb30}
\end{figure}
	\begin{figure}[!t]
\centering
\includegraphics[scale=0.60]{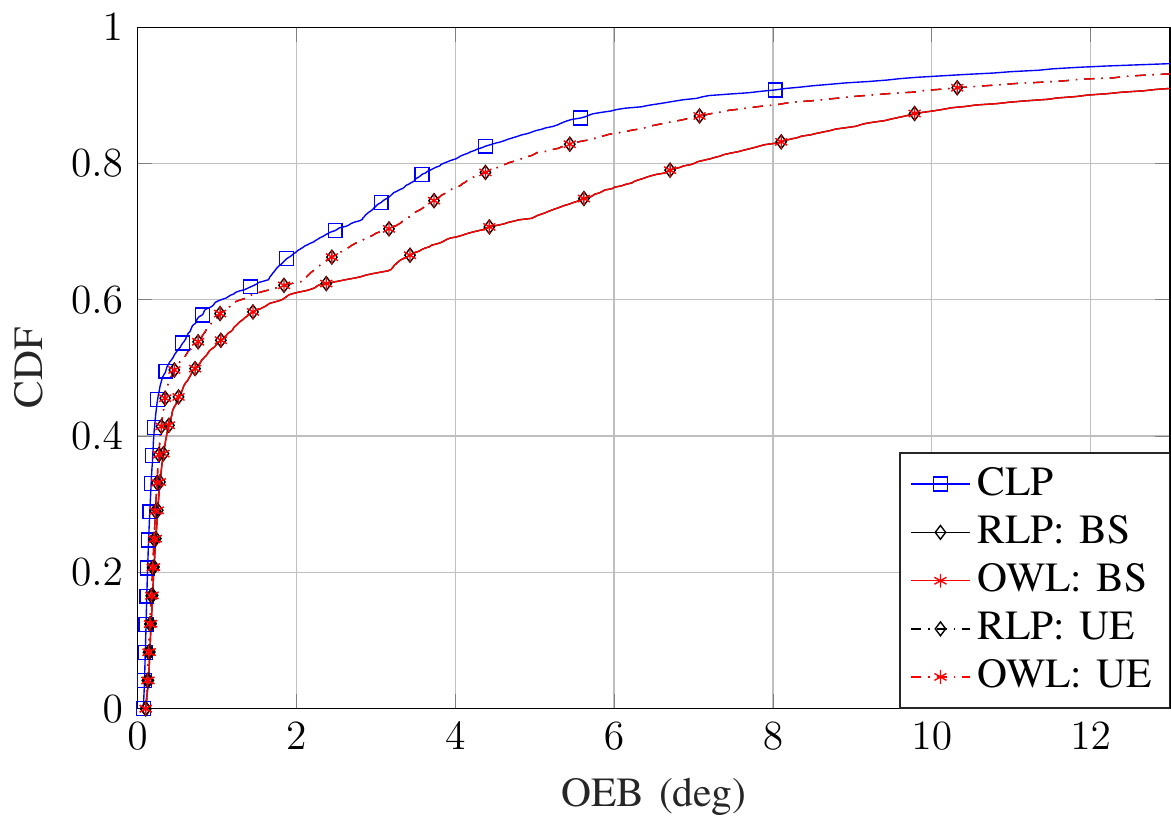}
\caption{CDF of OEB with UE orientation angles of $30^\circ$, and $N_\mathrm{UE}=N_\mathrm{BS}=144$, $N_\mathrm{B}=25$.}
\label{fig:oeb30}
	\end{figure}
The CDF of the PEB with orientation angles $\mathbf{o}=[30^\circ,30^\circ]^\mathrm{T}$ is shown in Fig.~\ref{fig:peb30}, for all the considered protocols. The overall observation from this figure, in comparison to Fig.~\ref{fig:peb0}, is that the performance worsens when the beams are steered away, i.e., when the orientation angles are non-zero. This can result in a loss of beamforming gain that depends non-linearly on the UE location and orientation angles. However, CLP performance is still superior to RLP and OWL. In this example, performance loss of $42$ cm, $54$ cm, and $80$ cm were observed at a PEB CDF of 90\%, under CLP, uplink RLP, and downlink RLP, respectively. On the other hand, comparing Fig.~\ref{fig:oeb30} with Fig.~\ref{fig:oeb0}, it can be seen that, at a CDF of 90\%, there is an OEB performance loss of $6.8^\circ$, $8.8^\circ$, and $11.5^\circ$ under CLP, uplink RLP, and downlink RLP, respectively. Considering the PEB and OEB loss, it can be concluded that, among the studied approaches, CLP is the approach that is most robust to UE mis-orientation. Finally, we note that in comparison to the case of matched orientation, under 30$^\circ$ mis-orientation, the system can still provide sub-meter PEB, while providing significantly higher OEB. This means that orientation estimation is more challenging than position estimation. {Recall that orientation changes the beamforming angles, which impacts localization performance. Hence, the study of orientation in this context is meaningful, despite this degraded performance.}

\subsection{Impact of the System Bandwidth on PEB}\label{sec:W}
\begin{figure}[!t]
	\centering
	\includegraphics[scale=0.60]{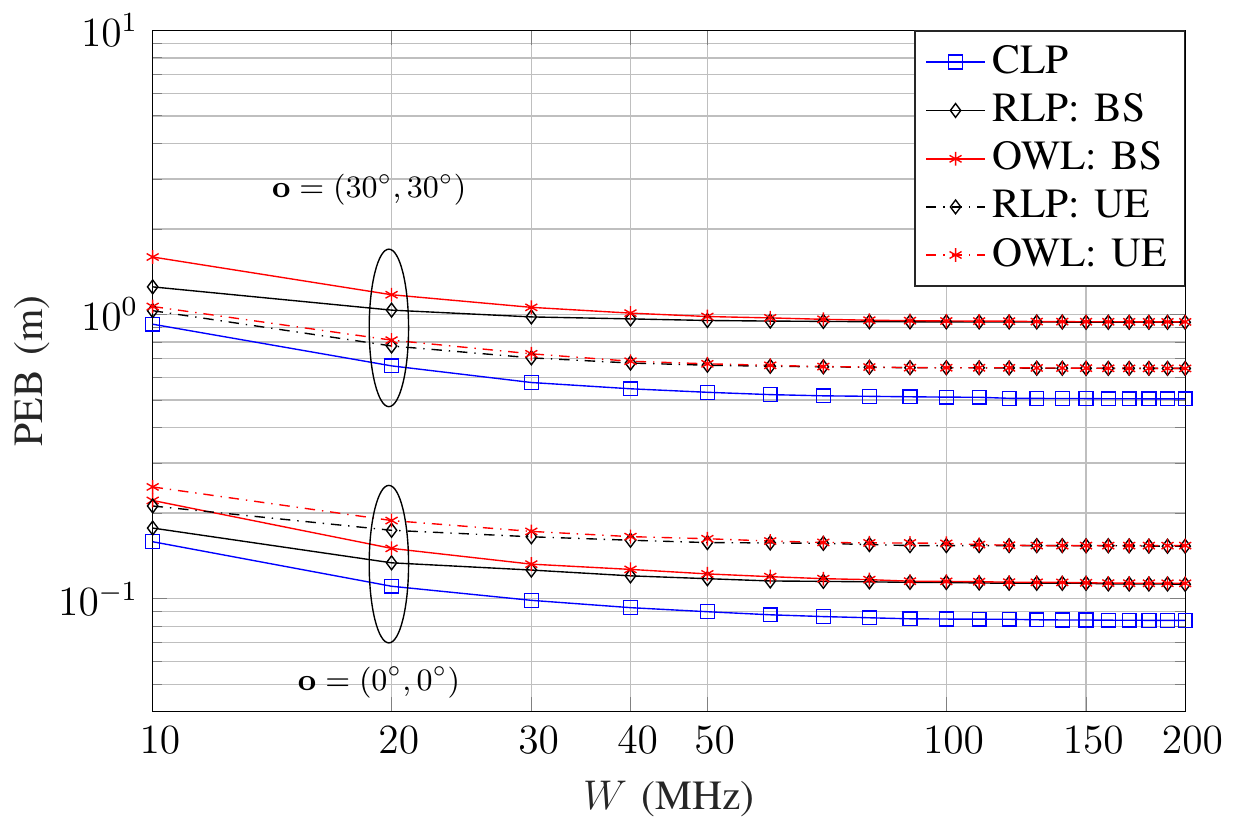}
	\caption{PEB at 0.9 CDF with respect to the bandwidth $W$.}
	\label{fig:W_sat}
\end{figure}
In Section \ref{sec:sim_peb_oeb}, we concluded that the system is limited by the estimation of the angles rather than the time delay. To investigate this phenomenon further, we now look closer into the impact of the bandwidth. {In the context of localization and ranging, higher bandwidths provide a more accurate estimation of the TOA, which leads to better localization bounds in general. Towards that,} the results in Fig.~\ref{fig:W_sat} indicate that as the bandwidth increases, the PEB decreases, until it reaches a floor at around $100$ MHz when $\mathbf{o}=[0^\circ,0^\circ]^\mathrm{T}$, and $60$ MHz when $\mathbf{o}=[30^\circ,30^\circ]^\mathrm{T}$. Based on these results, we make the following observations:
\begin{enumerate}
	\item At higher bandwidths that are more relevant in mmWave, the temporal information is very high compared to the spatial information, and the performance becomes fixed with $W$, i.e., the systems are spatially-limited.
	\item Under mis-orientation, the accuracy of spatial information degrades, and the system becomes spatially-limited. Hence, the improved temporal information does not provide any benefit to the performance achieved at lower bandwidths. 
	\item On the contrary, for lower bandwidths, the amount of temporal information decreases and becomes comparable to the spatial information. Therefore, the weight of the temporal information in the forward transmission becomes more significant, and the difference between OWL and RLP becomes more pronounced.
\end{enumerate}

\subsection{Impact of $N_\mathrm{BS}$ and $N_\mathrm{UE}$ on PEB}\label{sec:Nbs_Nue}
We now study the effect of the number of antennas at BS and UE on the PEB under CLP and RLP. Since this number can be $N_1$ or $N_2$ depending on the device role, we use $N_\mathrm{BS}$ and $N_\mathrm{UE}$ to unify the notation of the number of antennas at BS and UE, respectively.

Fig.~\ref{fig:N_ue} illustrates the effect of $N_\mathrm{UE}$ on PEB with $N_\mathrm{B}=25$ and $N_\mathrm{BS}=144$. It can be seen that at matched orientation (0$^\circ$, 0$^\circ$), performance tends to slightly improve with low to moderate $N_\mathrm{UE}$ values. However, higher $N_\mathrm{UE}$ generally results in a worse performance. This is because with higher $N_\mathrm{UE}$, the UE beams become narrower, and more beams are required to provide a full area coverage. It can also be noticed that, with an orientation of (30$^\circ$, 30$^\circ$), the rate of performance deterioration is higher. It is interesting to see that this rate is almost the same for the three protocols, which means that the performance loss is mainly due to SNR loss.
\begin{figure}[!t]
	\centering
	\includegraphics[scale=0.60]{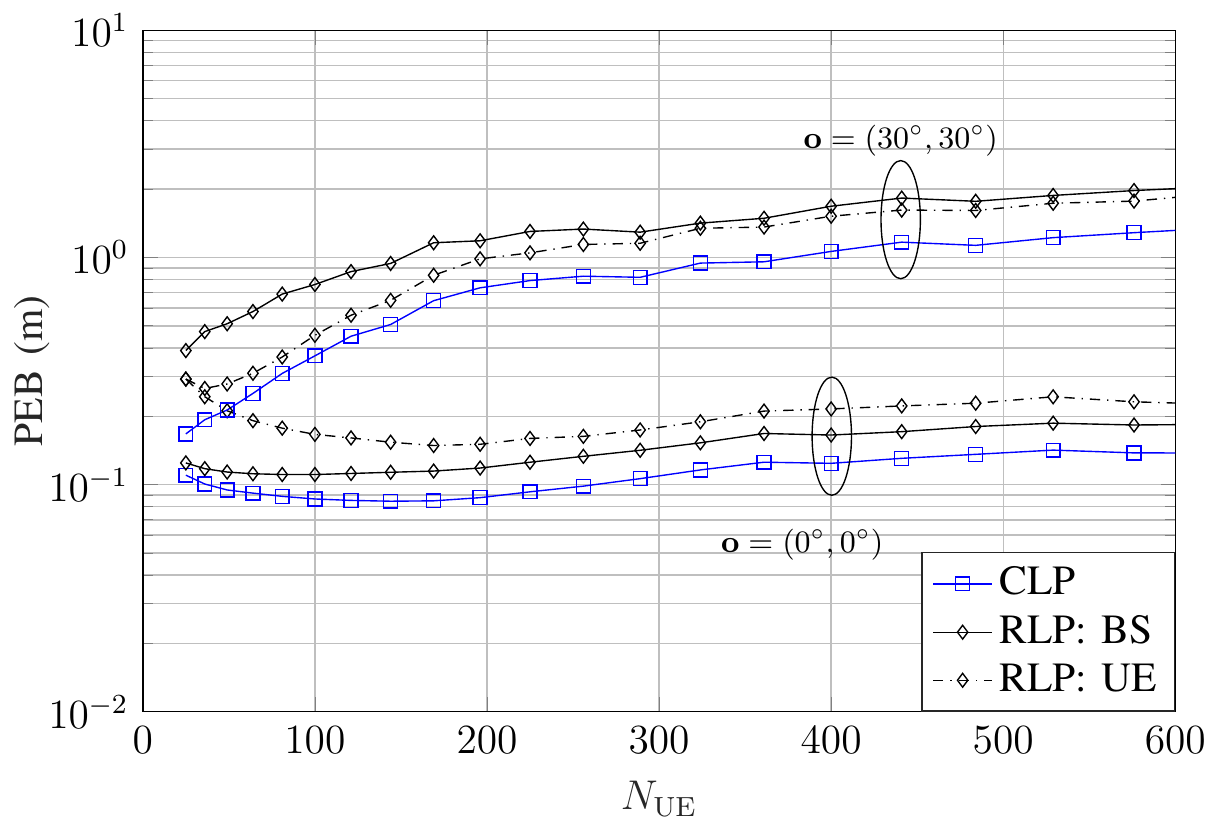}
	\caption{PEB at 0.9 CDF as a function of the UE number of antennas, with $N_\mathrm{B}=25$, with orientation angles $0^\circ$ and $30^\circ$, and $N_\mathrm{BS}=144$.}
	\label{fig:N_ue}
\end{figure}
\begin{figure}[!t]
	\centering
	\includegraphics[scale=0.60]{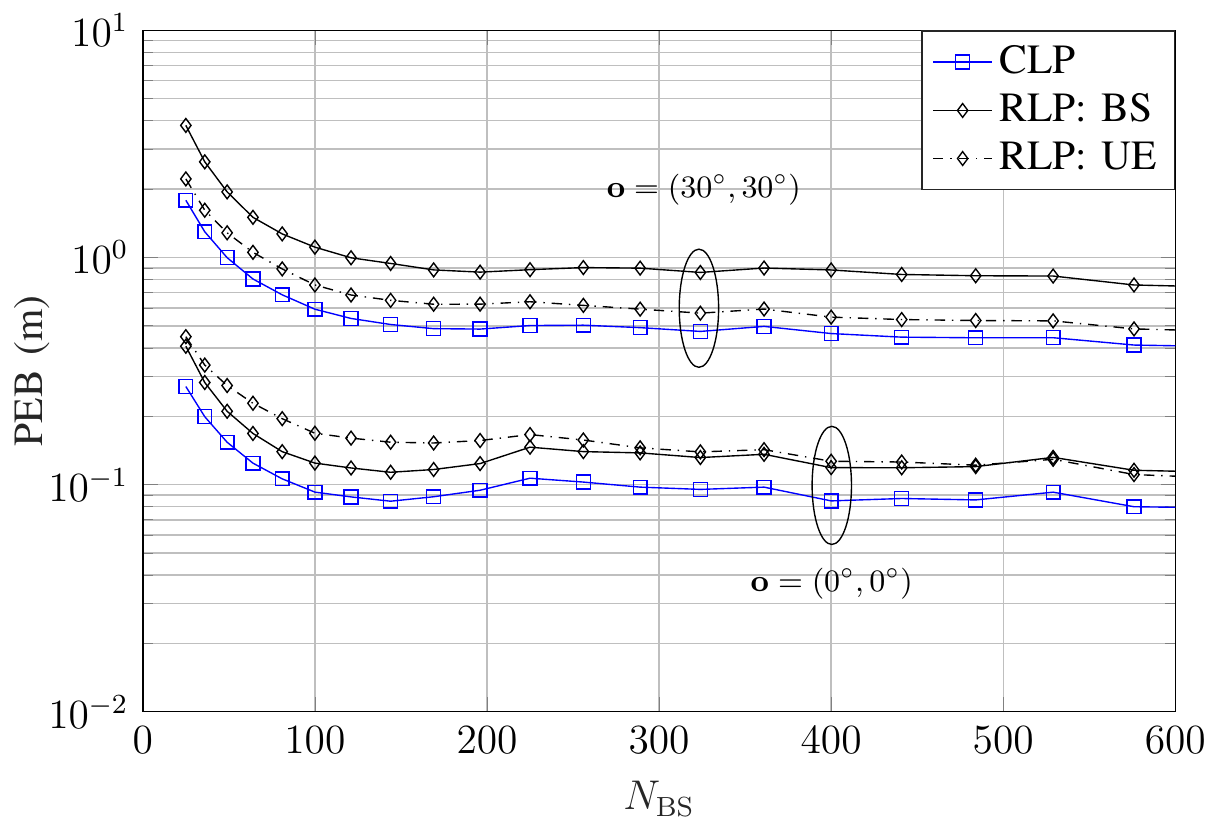}
	\caption{PEB at 0.9 CDF as a function of the BS number of antennas, with $N_\mathrm{B}=25$, with orientation angles $0^\circ$ and $30^\circ$, and $N_\mathrm{UE}=144$.}
	\label{fig:N_bs}
\end{figure}

On the other hand, the impact of $N_\mathrm{BS}$ is shown in Fig.~\ref{fig:N_bs} with $N_\mathrm{B}=25$ and $N_\mathrm{UE}=144$. It can be seen that a higher $N_\mathrm{BS}$ slightly improves the PEB in general. Similar to the case in Fig.~\ref{fig:N_ue}, it is understood that the PEB will generally increase when $N_\mathrm{BS}$ is arbitrarily large, albeit, at $N_\mathrm{BS}$ values well beyond those displayed in Fig.~\ref{fig:N_bs}, and with a lesser magnitude than higher $N_\mathrm{UE}$. Therefore, adding more antennas at the BS will not reduce the localization performance, as the UE antennas potentially would, at least within the studied range of array size. Finally, notice that both Figs.~\ref{fig:N_ue} and \ref{fig:N_bs} exhibit some non-monotonic trend. This is due to the nature of directional beamforming, whereby the beamforming gain depends on the user location, number of antennas, and beams directions as detailed in \cite{Zohair2016}. {In other words, varying the number of antennas results in a different sidelobe pattern that non-linearly varies the PEB and OEB.}

\section{Conclusions}\label{sec:conclusions}
Many publications on localization assume that the BS and UE are tightly synchronized. However, usually, communication systems are not synchronized to a high-level useful for localization. Focusing on this issue, in this paper, we considered two protocols of two-way localization referred to {round-trip localization protocol (RLP) and collaborative localization protocol (CLP)}. We investigated the PEB and OEB under these two protocols, where we showed mathematically that CLP outperforms RLP with a significant margin. However, this comes with the cost of requiring a feedback channel, unlike RLP where no synchronization or feedback are required, although it may need dedicated hardware to trigger the response. In our derivations, we considered beamforming at the transmitter and the receiver and accounted for the spatially-correlated receive noise. Considering the results of the numerical simulation, the enhancement observed for RLP over the traditional OWL was limited. That is, the localization was angle-limited rather than delay-limited. Moreover, our numerical results also showed that it is more beneficial to have more antennas at the BS than at the UE.

Future work based on this paper includes considering adaptive beamforming, whereby the beam directions are modified in the second round of transmission. Moreover, multipath propagation would be a relevant extension, since scatterers may differ in the uplink and downlink, depending on the beam directions. 

\section*{Acknowledgment}
The authors would like to thank Dr. Xiangyun Zhou of the Research School of Engineering at the Australian National University for his valuable feedback on this work.
\appendices
\section{Derivation of the Elements of the FIM of the Channel Parameters}\label{sec:app_fim_dl}
Consider backward transmission round. In this case, D$_1$ has the following observation
\begin{align}
\mathbf{y}_1(t)=\sqrt{N_\mathrm{1}N_\mathrm{2}E_\mathrm{t}}h^\mathrm{b}\mathbf{W}_\mathrm{1}^\mathrm{H}\mathbf{a}_\mathrm{1}\mathbf{a}_\mathrm{2}^\mathrm{T}\mathbf{F}_\mathrm{2}\mathbf{s}_2(t-\tau^\mathrm{b})+\mathbf{n}_1(t).
\end{align}
For the case of zero-mean additive correlated Gaussian noise, the FIM of $\boldsymbol{\varphi}^\mathrm{b}$ is given by \cite{kay1993}

	\begin{align}
	J^\mathrm{b}_{xy}\triangleq&
	=\frac{1}{N_0}\int_0^{T_\mathrm{o}}\Re\left\{\frac{\partial\boldsymbol{\mu}^\mathrm{H}(t)}{\partial x}\left(\mathbf{W}_1^\mathrm{H}\mathbf{W}_1\right)^{-1}\frac{\partial\boldsymbol{\mu}(t)}{\partial y}\right\}\mathrm{d}t,\\ &x,y\in\{\theta_{{1}},\phi_{{1}},\theta_{{2}},\phi_{{2}},\psi^\mathrm{b},\beta^\mathrm{b},\tau^\mathrm{b}\}\notag
	\end{align}

where $\boldsymbol{\mu}(t)$ is the mean of the observation vector, and $T_\mathrm{o}$ is assumed to be long enough to receive the entire pilot signal. Consequently, we write
\begin{align}
\boldsymbol{\mu}(t)=\sqrt{N_\mathrm{1}N_\mathrm{2}E_\mathrm{t}}\beta^\mathrm{b}\mathrm{e}^{j\psi^\mathrm{b}}\mathbf{W}_\mathrm{1}^\mathrm{H}\mathbf{a}_\mathrm{1}\mathbf{a}_\mathrm{2}^\mathrm{T}\mathbf{F}_\mathrm{2}\mathbf{s}_2(t-\tau^\mathrm{b}).
\end{align}
Defining  $\mathbf{\dot{s}}({t})\triangleq\frac{\partial\mathbf{{s}}({t})}{\partial t}, \mathbf{k}_\mathrm{i}=\frac{\partial}{\partial \theta_{i}}\mathbf{a}_\mathrm{i}, \mathbf{p}_\mathrm{i}=\frac{\partial}{\partial \phi_{i}}\mathbf{a}_\mathrm{i}$, $i\in\{1,2\}$, and the operator  $\mathbf{P}_\mathbf{A}\triangleq\mathbf{A}\left(\mathbf{A}^\mathrm{H}\mathbf{A}\right)^{-1}\mathbf{A}^\mathrm{H}$, and $\gamma\triangleq{N_1N_2N_\mathrm{s}E_\mathrm{t}}/{N_0}$, we can write the following
\begin{subequations}\label{eq:SpatialFIM1}
	\begin{align}
	J^\mathrm{b}_{\theta_1}&=\gamma{\beta^\mathrm{b}}^2\left(\mathbf{a}_\mathrm{2}^\mathrm{T}\mathbf{F}_\mathrm{2}\mathbf{F}_\mathrm{2}^\mathrm{H}\mathbf{a}_\mathrm{2}^{*}\right)\left(\mathbf{k}_\mathrm{1}^\mathrm{H}\mathbf{P}_{\mathbf{W}_1}\mathbf{k}_\mathrm{1}\right)\\
	J^\mathrm{b}_{\phi_1}&=\gamma{\beta^\mathrm{b}}^2\left(\mathbf{a}_\mathrm{2}^\mathrm{T}\mathbf{F}_\mathrm{2}\mathbf{F}_\mathrm{2}^\mathrm{H}\mathbf{a}_\mathrm{2}^{*}\right)\left(\mathbf{p}_\mathrm{1}^\mathrm{H}\mathbf{P}_{\mathbf{W}_1}\mathbf{p}_\mathrm{1}\right)\\
	J^\mathrm{b}_{\theta_2}&=\gamma{\beta^\mathrm{b}}^2\left(\mathbf{k}_\mathrm{2}^\mathrm{T}\mathbf{F}_\mathrm{2}\mathbf{F}_\mathrm{2}^\mathrm{H}\mathbf{k}_\mathrm{2}^{*}\right)\left(\mathbf{a}_\mathrm{1}^\mathrm{H}\mathbf{P}_{\mathbf{W}_1}\mathbf{a}_\mathrm{1}\right)\\
	J^\mathrm{b}_{\phi_2}&=\gamma{\beta^\mathrm{b}}^2\left(\mathbf{p}_\mathrm{2}^\mathrm{T}\mathbf{F}_\mathrm{2}\mathbf{F}_\mathrm{2}^\mathrm{H}\mathbf{p}_\mathrm{2}^{*}\right)\left(\mathbf{a}_\mathrm{1}^\mathrm{H}\mathbf{P}_{\mathbf{W}_1}\mathbf{a}_\mathrm{1}\right)\\
	J^\mathrm{b}_{\beta^\mathrm{b}}&=\gamma\left(\mathbf{a}_\mathrm{2}^\mathrm{T}\mathbf{F}_\mathrm{2}\mathbf{F}_\mathrm{2}^\mathrm{H}\mathbf{a}_\mathrm{2}^{*}\right)\left(\mathbf{a}_\mathrm{1}^\mathrm{H}\mathbf{P}_{\mathbf{W}_1}\mathbf{a}_\mathrm{1}\right),\\
	J^\mathrm{b}_{\psi^\mathrm{b}}&=\gamma{\beta^\mathrm{b}}^2\left(\mathbf{a}_\mathrm{2}^\mathrm{T}\mathbf{F}_\mathrm{2}\mathbf{F}_\mathrm{2}^\mathrm{H}\mathbf{a}_\mathrm{2}^{*}\right)\left(\mathbf{a}_\mathrm{1}^\mathrm{H}\mathbf{P}_{\mathbf{W}_1}\mathbf{a}_\mathrm{1}\right),\\
	J^\mathrm{b}_{\theta_1\phi_1}&=\gamma{\beta^\mathrm{b}}^2\left(\mathbf{a}_\mathrm{2}^\mathrm{T}\mathbf{F}_\mathrm{2}\mathbf{F}_\mathrm{2}^\mathrm{H}\mathbf{a}_\mathrm{2}^{*}\right)\left(\mathbf{p}_\mathrm{1}^\mathrm{H}\mathbf{P}_{\mathbf{W}_1}\mathbf{k}_\mathrm{1}\right),\label{eq:example}\\
	J^\mathrm{b}_{\theta_1\theta_2}&=\gamma{\beta^\mathrm{b}}^2\left(\mathbf{k}_\mathrm{2}^\mathrm{T}\mathbf{F}_\mathrm{2}\mathbf{F}_\mathrm{2}^\mathrm{H}\mathbf{a}_\mathrm{2}^{*}\right)\left(\mathbf{k}_\mathrm{1}^\mathrm{H}\mathbf{P}_{\mathbf{W}_1}\mathbf{a}_\mathrm{1}\right),\\
	J^\mathrm{b}_{\theta_1\phi_2}&=\gamma{\beta^\mathrm{b}}^2\left(\mathbf{p}_\mathrm{2}^\mathrm{T}\mathbf{F}_\mathrm{2}\mathbf{F}_\mathrm{2}^\mathrm{H}\mathbf{a}_\mathrm{2}^{*}\right)\left(\mathbf{k}_\mathrm{1}^\mathrm{H}\mathbf{P}_{\mathbf{W}_1}\mathbf{a}_\mathrm{1}\right),\\
	J^\mathrm{b}_{\theta_1\beta^\mathrm{b}}&=\gamma{\beta^\mathrm{b}}\left(\mathbf{a}_\mathrm{2}^\mathrm{T}\mathbf{F}_\mathrm{2}\mathbf{F}_\mathrm{2}^\mathrm{H}\mathbf{a}_\mathrm{2}^{*}\right)\left(\mathbf{k}_\mathrm{1}^\mathrm{H}\mathbf{P}_{\mathbf{W}_1}\mathbf{a}_\mathrm{1}\right)\\
	J^\mathrm{b}_{\phi_1\theta_2}&=\gamma{\beta^\mathrm{b}}^2\left(\mathbf{k}_\mathrm{2}^\mathrm{T}\mathbf{F}_\mathrm{2}\mathbf{F}_\mathrm{2}^\mathrm{H}\mathbf{a}_\mathrm{2}^{*}\right)\left(\mathbf{p}_\mathrm{1}^\mathrm{H}\mathbf{P}_{\mathbf{W}_1}\mathbf{a}_\mathrm{1}\right),\\
	J^\mathrm{b}_{\phi_1\phi_2}&=\gamma{\beta^\mathrm{b}}^2\left(\mathbf{p}_\mathrm{2}^\mathrm{T}\mathbf{F}_\mathrm{2}\mathbf{F}_\mathrm{2}^\mathrm{H}\mathbf{a}_\mathrm{2}^{*}\right)\left(\mathbf{p}_\mathrm{1}^\mathrm{H}\mathbf{P}_{\mathbf{W}_1}\mathbf{a}_\mathrm{1}\right),\\
	J^\mathrm{b}_{\phi_1\beta^\mathrm{b}}&={\gamma}{\beta^\mathrm{b}}\left(\mathbf{a}_\mathrm{2}^\mathrm{T}\mathbf{F}_\mathrm{2}\mathbf{F}_\mathrm{2}^\mathrm{H}\mathbf{a}_\mathrm{2}^{*}\right)\left(\mathbf{p}_\mathrm{1}^\mathrm{H}\mathbf{P}_{\mathbf{W}_1}\mathbf{a}_\mathrm{1}\right),\\
	J^\mathrm{b}_{\theta_2\phi_2}&=\gamma{\beta^\mathrm{b}}^2\left(\mathbf{p}_\mathrm{2}^\mathrm{T}\mathbf{F}_\mathrm{2}\mathbf{F}_\mathrm{2}^\mathrm{H}\mathbf{k}_\mathrm{2}^{*}\right)\left(\mathbf{a}_\mathrm{1}^\mathrm{H}\mathbf{P}_{\mathbf{W}_1}\mathbf{a}_\mathrm{1}\right),\\
	J^\mathrm{b}_{\theta_2\beta^\mathrm{b}}&={\gamma}{\beta^\mathrm{b}}\left(\mathbf{a}_\mathrm{2}^\mathrm{T}\mathbf{F}_\mathrm{2}\mathbf{F}_\mathrm{2}^\mathrm{H}\mathbf{k}_\mathrm{2}^{*}\right)\left(\mathbf{a}_\mathrm{1}^\mathrm{H}\mathbf{P}_{\mathbf{W}_1}\mathbf{a}_\mathrm{1}\right),\\
	J^\mathrm{b}_{\phi_2\beta^\mathrm{b}}&={\gamma}{\beta^\mathrm{b}}\left(\mathbf{a}_\mathrm{2}^\mathrm{T}\mathbf{F}_\mathrm{2}\mathbf{F}_\mathrm{2}^\mathrm{H}\mathbf{p}_2^{*}\right)\left(\mathbf{a}_\mathrm{1}^\mathrm{H}\mathbf{P}_{\mathbf{W}_1}\mathbf{a}_\mathrm{1}\right),\\
	J_{\tau^\mathrm{b}}&=4\gamma{\beta^\mathrm{b}}^2\pi^2W_\mathrm{eff}^2\left(\mathbf{a}_\mathrm{2}^\mathrm{T}\mathbf{F}_\mathrm{2}\mathbf{F}_\mathrm{2}^\mathrm{H}\mathbf{a}_\mathrm{2}^{*}\right)\left(\mathbf{a}_\mathrm{1}^\mathrm{H}\mathbf{P}_{\mathbf{W}_1}\mathbf{a}_\mathrm{1}\right).  \label{eq:tau_b_tau_f}
	\end{align}
\end{subequations}
\normalsize
where
\begin{align*}
W_\mathrm{eff}^2&=\int_{-W/2}^{W/2}f^2 |G(f)|^2\mathrm{d}f.\notag
\end{align*}
Other entries in $\mathbf{J}_{\boldsymbol{\varphi}}^{\mathrm{b}}$ are zero because
\begin{align}
\int_0^{T_\mathrm{o}}\mathbf{{s}}_2^\mathrm{H}(t-\tau^\mathrm{b})\mathbf{\dot{s}}_2(t-\tau^\mathrm{b})\mathrm{d}t&=0,\\
\int_0^{T_\mathrm{o}}\mathbf{{s}}_2(t-\tau^\mathrm{b})\mathbf{{s}}_2^\mathrm{H}(t-\tau^\mathrm{b})\mathrm{d}t&=N_\mathrm{s}\mathbf{I}_{N_\mathrm{B}}.
\end{align}

In forward transmission, the subscripts  $``1"$ and $``2"$ should be interchanged in \eqref{eq:SpatialFIM1} and the superscript $``\mathrm{b}"$ replaced by $``\mathrm{f}"$. For example, from $J^\mathrm{b}_{\theta_1\phi_1}$ in \eqref{eq:example}, we can calculate $J^\mathrm{f}_{\theta_2\phi_2}=\gamma{\beta^\mathrm{f}}^2\left(\mathbf{a}_\mathrm{1}^\mathrm{T}\mathbf{F}_\mathrm{1}\mathbf{F}_\mathrm{1}^\mathrm{H}\mathbf{a}_\mathrm{1}^{*}\right)\left(\mathbf{p}_\mathrm{2}^\mathrm{H}\mathbf{P}_{\mathbf{W}_2}\mathbf{k}_\mathrm{2}\right)$, which goes in row 3, column 4 in the forward-transmission counterpart of \eqref{eq:jthetatheta}.

\section{Proof of Theorem 1}\label{app:proof_theo}
Define the vector of all unknown parameters as $\mathbf{v}=[\mathbf{x}^\mathrm{T},\mathbf{z}_1^\mathrm{T},\mathbf{z}_2^\mathrm{T}]^\mathrm{T}$, then the FIM of $\mathbf{v}$ based on the first and second observations are, respectively,
\begin{align*}
    \mathbf{J}_\mathbf{v}^{(1)}&=\begin{bmatrix}
    \mathbf{J}_\mathbf{x}^{(1)}&\mathbf{J}_{\mathbf{x},\mathbf{z}_1}^{(1)}&\mathbf{0}\\ 
    \mathbf{J}_{\mathbf{x},\mathbf{z}_1}^{\mathrm{T}(1)}&\mathbf{J}_{\mathbf{z}_1}^{(1)}&\mathbf{0}\\
    \mathbf{0}&\mathbf{0}&\mathbf{0}
    \end{bmatrix},\ 
    \mathbf{J}_\mathbf{v}^{(2)}=\begin{bmatrix}
    \mathbf{J}_\mathbf{x}^{(2)}&\mathbf{0}&\mathbf{J}_{\mathbf{x},\mathbf{z}_2}^{(2)}\\ 
    \mathbf{0}&\mathbf{0}&\mathbf{0}\\
    \mathbf{J}_{\mathbf{x},\mathbf{z}_2}^{\mathrm{T}(2)}&\mathbf{0}&\mathbf{J}_{\mathbf{z}_2}^{(2)}
    \end{bmatrix}
\end{align*}
Since the two observations are independent,
\begin{align}
    \mathbf{J}_\mathbf{v}&=\begin{bmatrix}
    \mathbf{J}_\mathbf{x}^{(1)}+\mathbf{J}_\mathbf{x}^{(2)}&\mathbf{J}_{\mathbf{x},\mathbf{z}_1}^{(1)}&\mathbf{J}_{\mathbf{x},\mathbf{z}_2}^{(2)}\\ 
    \mathbf{J}_{\mathbf{x},\mathbf{z}_1}^{\mathrm{T}(1)}&\mathbf{J}_{\mathbf{z}_1}^{(1)}&\mathbf{0}\\
       \mathbf{J}_{\mathbf{x},\mathbf{z}_2}^{\mathrm{T}(2)}&\mathbf{0}&\mathbf{J}_{\mathbf{z}_2}^{(2)}
    \end{bmatrix}
\end{align}
Consequently, EFIM of $\mathbf{x}$ is given by Schur complement as 
\begin{align}
    \mathbf{J}_\mathbf{x}^\mathrm{e}=& \mathbf{J}_\mathbf{x}^{(1)}+\mathbf{J}_\mathbf{x}^{(2)}\notag\\&-\mathbf{J}_{\mathbf{x},\mathbf{z}_1}^{(1)}\left(\mathbf{J}_{\mathbf{z}_1}^{(1)}\right)^{-1}\mathbf{J}_{\mathbf{x},\mathbf{z}_1}^{\mathrm{T}(1)}-\mathbf{J}_{\mathbf{x},\mathbf{z}_2}^{(2)}\left(\mathbf{J}_{\mathbf{z}_2}^{(2)}\right)^{-1}\mathbf{J}_{\mathbf{x},\mathbf{z}_2}^{\mathrm{T}(2)}\label{eq:proof}
\end{align}
Note that the first and third term in \eqref{eq:proof} represent the Schur complement of $\mathbf{x}$ with respect to $\mathbf{z}_1$ obtained from the first process, while the second and fourth term represent the Schur complement of $\mathbf{x}$ with respect to $\mathbf{z}_2$ obtained from the second process. In other words,
\begin{align}
        \mathbf{J}^\mathrm{e}_\mathbf{x}=\mathbf{J}^\mathrm{e,1}_\mathbf{x}+\mathbf{J}^\mathrm{e,2}_\mathbf{x}.
\end{align}
\bibliographystyle{IEEEtran}

	\end{document}